\documentclass[11pt, a4paper]{article}
\usepackage{amsthm, amssymb, amsmath}
\usepackage{url}
\usepackage{array}
\usepackage{subcaption}
\usepackage{mathtools}
\usepackage{multicol}
\usepackage{amsmath}
\usepackage{amssymb}
\usepackage{xcolor}

\usepackage[utf8]{inputenc} 
\usepackage[T1]{fontenc}    
\usepackage{hyperref}       
\usepackage{url}            
\usepackage{booktabs}       
\usepackage{amsfonts}       
\usepackage{nicefrac}       
\usepackage{microtype}      
\usepackage[absolute]{textpos}
\usepackage{algorithm,algorithmic}
\usepackage{color}
\usepackage{bm}
\usepackage{float}
\usepackage{subfloat}
\usepackage{multirow}
\usepackage{diagbox}
\usepackage{cite}
\usepackage{tabu}
\usepackage{hyperref}
\usepackage{array}
\usepackage[top=3cm, bottom=3cm, left=3cm, right=3cm]{geometry}
\theoremstyle{definition}
 
\newtheorem{theorem}{Theorem}[section]

\newtheorem{corollary}[theorem]{Corollary}

\theoremstyle{remark}

\renewcommand{\subset}{\subseteq}

\newcommand{\st}{\colon\,}

\newcommand{\caze}[2]{\textbf{Case {#1}:} \textit{#2}}
\newcommand{\sizeof}[1]{\left\lvert{#1}\right\rvert}

\newcommand{\kay}{\mathcal{K}}

\begin{document}

\title{Motif and Hypergraph Correlation Clustering}
\author{ Pan~Li \and Gregory J.~Puleo \and Olgica Milenkovic}
\footnotetext[1]{Pan Li and Olgica Milenkovic are with the Coordinated Science Laboratory, Department of Electrical and Computer Engineering, University of Illinois at Urbana-Champaign (email: panli2@illinois.edu, milenkov@illinois.edu)}
\footnotetext[2]{Gregory J.~Puleo are with Department of Mathematics, Auburn University. (email: gjp0007@auburn.edu)}
\footnotetext[3]{This work was supported in part by NSF Grants CIF 1218764, CIF 1117980, 1339388 and STC Class 2010, CCF 0939370. Research of the second author was partly supported by the IC Postdoctoral Research Fellowship. Part of the results were presented at the Infocom 2017 conference~\cite{li2017motif}.}
\date{}
\maketitle
\renewcommand*{\thefootnote}{\arabic{footnote}}

\begin{abstract}
Motivated by applications in social and biological network analysis, we introduce a new form of agnostic clustering termed~\emph{motif correlation clustering}, which aims to minimize the cost of clustering errors associated with both edges and higher-order network structures. The problem may be succinctly described as follows: Given a complete graph $G$, partition the vertices of the graph so that certain predetermined ``important'' subgraphs mostly lie within the same cluster, while ``less relevant'' subgraphs are allowed to lie across clusters. 
Our contributions are as follows: We first introduce several variants of motif correlation clustering and then show that these clustering problems are NP-hard. We then proceed to describe polynomial-time clustering algorithms that provide constant approximation guarantees for the problems at hand. Despite following the frequently used LP relaxation and rounding procedure, the algorithms involve a sophisticated and carefully designed neighborhood growing step that combines information about both edge and motif structures. We conclude with several examples illustrating the performance of the developed algorithms on synthetic and real networks. 
\end{abstract}

\section{Introduction}
Correlation clustering is a clustering model first introduced by Bansal, Blum, and Chawla in~\cite{bansal2004correlation}
and it may be succinctly described as follows:  One is given a collection of objects and, for some pairs of objects, one is also given quantitative assessments of whether the objects
are \emph{similar} or \emph{dissimilar}. This information is represented using a labeled graph with edges marked by $+$ or $-$ symbols 
according to whether the endpoints are similar or dissimilar. The goal
is to partition the vertices of the graphs so that edges labeled by $+$ tend to aggregate within clusters
and edges labeled by $-$ tend to go across clusters. Unlike most other known clustering methods, 
correlation clustering does not require the number of clusters to be specified in advance. 

There are two formulations of the correlation clustering optimization problem: \emph{MinDisagree} and \emph{MaxAgree}. In the MinDisagree version of the problem one aims to minimize the number of erroneously placed edges, while in the MaxAgree version one seeks to maximize the total number of correctly placed edges. Finding an optimal solution 
to either problem is NP-hard. The MinDisagree problem remains hard even when the input graph is
complete~\cite{bansal2004correlation}.  
For complete graphs, several constant approximation randomized~\cite{ailon2008aggregating} and deterministic~\cite{charikar2003clustering} algorithms are known. When the graph is allowed to be arbitrary, the best known approximation ratio is $O(\log n)$~\cite{charikar2003clustering}.
Although finding an optimal solution
for the MaxAgree problem is hard, approximating an optimal solution is in this case significantly easier than for the case of MinDisagree~\cite{bansal2004correlation}.

Several variants of correlation clustering allow for including edge weights into the problem formulation, with each edge endowed with a ``similarity'' and ``dissimilarity'' weight: If the edge is placed across clusters, the edge is charged its similarity cost, and if the edge is placed within the same cluster, the edge is charged its dissimilarity cost. The MinDisagree clustering goal is to minimize the overall vertex partitioning cost. Clearly, if the weights are unrestricted, not all instances of the weighted clustering problem may be efficiently approximated. Hence, most of the work has focused on so-called \emph{probability weights}~\cite{bansal2004correlation}.
 
We depart from classical correlation clustering problems by considering a new setting in which one is allowed to assign probability weights to both edges and arbitrary small induced subgraphs in the graph (e.g., triangles) and then perform the clustering so as to minimize the overall cost of both edge and motif placements or motif placements alone. This enables one to extend traditional correlation clustering by considering higher-order structures in the network such as paths, triangles and cycles. Given that subgraphs/motifs may be modeled as hyperedges in a hypergraph, our line of work complements recent works on spectral hypergraph clustering methods~\cite{chien2018community,li2017inhomogeneous,li2018submodular,lin2017fundamental,ghoshdastidar2014consistency} and heuristic tensor spectral clustering methods~\cite{benson2015tensor}, as well as other generalizations of correlation clustering~\cite{puleo2015correlation,puleo2018correlation}. Furthermore, the proposed method allows for handling motifs in directed graphs by converting the directed graphs into undirected graphs while retaining information about the ``relevance'' of directed subgraphs within the graph. This relevance information may be incorporated into similarity and dissimilarity weights (for example, if only feedforward triangle motifs are relevant, only those directed motifs will be assigned large weight in the undirected graph and hence encouraged to fall within one cluster). As a result, motif clustering may have widespread applications for discovery of layered flows in a information networks, anomaly detection in communication networks or for determining hierarchical community structure detection in gene regulatory networks~\cite{benson2015tensor,benson2016higher}. In particular, motif correlation clustering can lead to pathway and community recovery for graphs and networks where higher order structures carry significantly more relevant information about the functionality, direction and strength of connections in the network than edges alone.

Our contributions are as follows. We rigorously formulate the
first known MinDisagree motif correlation clustering problem, and show
that it is NP-hard even when only some special motifs such as triangles are considered. 
Next, we introduce an extended correlation
clustering framework which allows to both fine tune the cost of
clustering edges and higher order network structures. 
We then describe a new
two-stage clustering algorithm comprising an LP and rounding step, and
show that the algorithm offers constant approximation guarantees that
depend on the size of the motifs under consideration. We also provide
examples illustrating that standard randomized pivoting algorithms
fail on this instance of correlation clustering. 
Our expositions concludes with several examples
pertaining to synthetic and real network analysis, such social, flow, and anomaly detection networks, illustrating the advantages of motif correlation clustering over other edge-based methods.

Since the publication of our original preliminary findings on the topic of motif correlation clustering~\cite{li2017motif}, follow up work was reported in~\cite{gleich2018correlation}. There, several versions of our motif clustering methods have been adapted to yield improved approximation constants. Nevertheless, new results reported in our work offer both new sophisticated proof techniques, mixed motif models as well as approximation guarantees that improve those reported in~\cite{gleich2018correlation}.

\section{Notation and problem formulation} \label{sec:intro} 
Let $G(V,E)$ be a complete, undirected graph with vertex set $V$ of cardinality $n$ and edge set $E$ of cardinality $\binom{n}{2}$. For simplicity, we will assume that the vertices are endowed with distinct integer labels in $[n]$ and this labeling introduces a natural ordering of the vertices. Also, we let $C=(C_1,\ldots,C_s)$, $1 \leq s \leq n,$ stand for a partition of the vertex set $[v]$ and $\mathcal{C}_n$ for the set of all partitions of $[n]$.  

Let $S$ be a (sub)set of vertices and let $\mathcal{K}(S)$ denote the set of all $k$-subsets of vertices in $S$, where $2 \leq k<n$ is a constant independent of $n$. Clearly, $|\mathcal{K}(S)|={|S| \choose k}$ and $|\mathcal{K}(V)|={n \choose k}$. Denote a subgraph of $G$ induced by a $k$-subset of vertices by $K_k\in \mathcal{K}(V)$ (whenever clear from the context, we omit the subscript $k$). Each $K$ is associated with a pair of non-negative values $(w_K^+,w_K^-)$. The weights $w_{K}^+$ and $w_{K}^-$ indicate the respective costs of placing this $k$-tuple $K$ across and within the cluster, respectively. Note that in most practical settings, the most relevant motifs in a graph are edges and triangles.
Typically, in practice, the $k$-tuple $K$ corresponds to a graph motif and its corresponding weights are determined by the functionality of that motif. For transparency of notation, we denote variables $x$ associated with $K$-subsets by $x_{K}$  or edges by $uv,\, u,v \in V$ and $ x_{uv}=x_{vu}$.

Our goal is to solve two MinDisagree versions of the problem: In the first version of the problem, termed \emph{motif correlation clustering (MCC)}, we first fix one motif graph on $k$ vertices and then seek a vertex partition $C\in\mathcal{C}_n$, that minimizes the objective function:
\begin{align}\label{eq:mcc}
\text{(MCC)} \;\;\;\;\; \min_{C \in \mathcal{C}_n}\sum_{K \subseteq C_i, \, \text{ for some $i$}}w_{K}^-+\sum_{K\not \subseteq C_i, \, \text{for all $i$}}w_{K}^+\,.
\end{align}
In the second version of the problem, termed \emph{mixed motif correlation clustering (MMCC)}, we are allowed to fix multiple motif graphs of possibly different sizes $2\leq k_1 < k_2 < \ldots < k_p,$
 and we seek a vertex partition $C=(C_1,\ldots,C_s)$, $s \geq 1$, that minimizes the objective function:
\begin{align}\label{eq:mmcc}
\text{(MMCC)} \;\;\;\;\; \min_{C \in \mathcal{C}_n} \sum_{t=1}^p \;  \lambda_{t} \, \left( \sum_{K\subseteq C_i, \, \text{ for some $i$}, \, |K|=k_t}w_{K}^-+\sum_{K\not \subseteq C_i, \, \text{for all $i$}, \, |K|=k_t}w_{K}^+ \right) \,.
\end{align}

Here, $\lambda_{t} \geq 0$ are relevance factors of the motifs of size $k_t$. Note that by choosing $\lambda=1$ for edges and setting all other relevance factors to zero, we arrive at the classical correlation clustering formulation. Furthermore, in both problems, we impose the triangle constraint on the weights $w_K^+ + w_K^-=1$.

Clearly, both the MCC and MMCC problems are NP-complete, as the correlation clustering problem is NP-complete. Furthermore, the following theorem, proved in Appendix~\ref{app1}, shows that the problems remain hard even for restricted choices of motifs, such as the case when $k=3$. Hence, we focus on developing (constant) approximation algorithms for the problems.
\begin{theorem}\label{thm:hardness}
For $k=3$, the MCC problem is NP-complete.
\end{theorem}

We pointe out that one may also consider the MaxAgree version of the motif clustering problem where the objective functions in (MCC) and (MMCC) that summarize disagreement are replaced by objective functions that summarize agreement, and where correspondingly the $\min$ function is replaced by the $\max$ function. As for the case of correlation clustering, it is straightforward to show that taking the better of two clusterings, the all-singleton clustering and the single-component clustering provides a $2$-approximation for the problem. 

The MinDisagree version of correlation clustering is usually approximately solved using two approaches: Pivoting methods~\cite{ailon2008aggregating} and relaxed Integer Programming (IP) methods that reduce to solving a Linear Program (LP) followed by rounding~\cite{charikar2003clustering}. The pivoting algorithm is a straightforward randomized approach that provides constant approximation guarantees for the expected value of the objective, and has straightforward, yet efficient, 
parallel implementations~\cite{pan2015parallel}. For the unweighted clustering problem, it may be succinctly described as follows: One selects a pivot vertex uniformly at random, incorporates all its ``similar'' neighbors (i.e., those with edge label `+') into one cluster, removes all vertices in the newly formed cluster from the graph and then proceeds to iteratively repeat the same steps. Unfortunately, using this approach for motif clustering cannot lead to constant approximation results, as illustrated by the example below. 

Consider the MCC problem for complete graphs and triple-motifs, i.e., for $k=3$. Suppose that each edge is labeled, with labels in the set $\{{+,-\}}$, and that each triple $K$ is associated with a pair of weights $(w_{K}^+,w_{K}^-)\in\{(1,0),(0,1)\}$. triples that correspond to triangles with positively labeled edges only have weights $(w_{K}^+,w_{K}^-)=(1,0)$, and are termed ``positive'' triples. All other triples have weights $(w_{K}^+,w_{K}^-)=(0,1)$, and are termed ``negative'' triples. For this setting, neither pivoting on a pair of vertices (e.g., an edge) nor pivoting on a single vertex may provide constant approximation guarantees, as demonstrated by the examples in Figure~\ref{pivoting}. Both graphs are complete graphs but for ease of interpretation, only positively labeled edges are depicted. In the first case, one chooses a (positive) edge uniformly at random and includes in the cluster all positive edges connected to the pivoting edge. For Figure~\ref{pivoting} a), the optimal clustering comprises two clusters, $C_1=\{v_1,v_2,v_3\}$ and $C_2=\{v_4,v_5,v_6\}$ and has an MCC objective function value equal to zero. If one pivots on the edge $(v_1,v_4)$, the resulting clustering contains one cluster only, $C1=\{v_1,v_2,\ldots,v_6\}$, and leads to a positive value of the objective function, and hence an unbounded ratio of the optimal and approximate objective. Pivoting on vertices may fail as well, which may be seen from example b): The graph in b) has a unique optimal clustering with two clusters $C_1=\{v_1,v_2,v_3\}$ and $C_2=\{v_4,v_5,...,v_n\}$. Choosing the vertex $v_3$ as pivot and including all vertices connected to $v_3$ through positive edges leads to $v_1,v_2, v_4$ being clustered together with $v_3$, thereby resulting in $O(n^2)$ more errors than those incurred by the optimal clustering. As there are $n$ vertices in the graph, the expected value of the objective may have an error term $O(n)$.

\begin{figure}[h] 
\centering
\includegraphics[trim={0cm 19.6cm 1cm 0.9cm}, clip, width=0.9\linewidth]{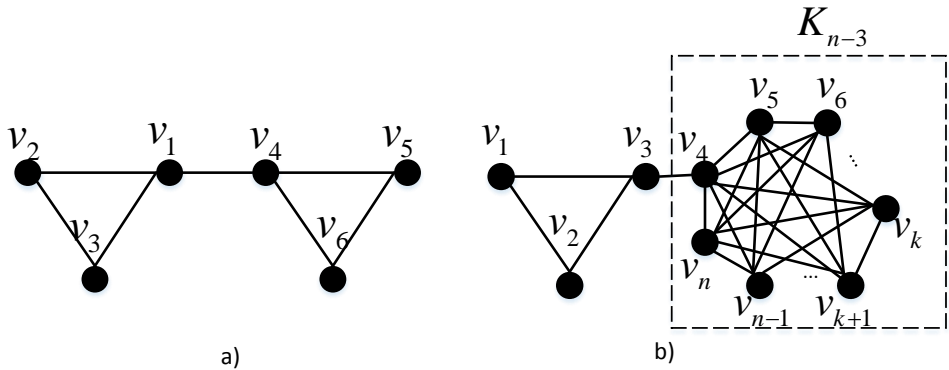} 
\caption{Pivoting on edges and vertices of graphs. Note that $K_{n-3}$ stands for a complete graph on $n-3$ vertices.}
\label{pivoting}
\end{figure} 

\section{Main results} \label{sec:main} 

We describe next polynomial-time, constant approximation algorithms for the MCC and MMCC problems. For the former case, we propose two methods that offer different trade-offs between optimization performance and complexity, as measured in terms of the number of constraints used in the underlying LP program. The approach followed is to relax the IPs of~(\ref{eq:mcc}) and~(\ref{eq:mmcc})  to LPs and then perform rounding of the fractional solutions. The main analytical difficulties encountered in this approach are that the LPs involve both edge and higher order motif variables, and that trying to round all these variables simultaneously may cause inconsistencies and large rounding errors. More precisely, in the LP formulation one has to incorporate variables associated with $k$-tuples, while rounding only works with variables associated with pairs of vertices. To overcome this issue for the MCC problem, our first solution introduces motif variables in the LP and then performs rounding on edges by assigning to them a cost that reflects the value of the best-scoring motif that the edge is part of. The second solution is based on an LP which involves both motif and edge variables and allows downstream rounding to be directly performed on the edge variables. The second method has fewer constraints in the underlying LP than the first method, and is hence more computationally efficient. The drawback is that provides worse approximation guarantees than the first method. For the MMCC problem, one may use the second method developed for the MCC problem with the inclusion of additional constraints for $k$-tuple and edge variables. The approximation factor is determined by the size of largest motifs. 

As in the formulation of the MCC problem, let $K$ correspond to a $k$-tuple and let $x_K$ denote the indicator variable for 
the event that the vertices in $K$ are split among clusters (i.e., $x_K=0$ if the vertices of $K$ lie in the same cluster, and $x_K=1$ otherwise). 
Relaxing the above integral constraint to $x_K \in [0,1] $ and rewriting the probability weight constraints leads to the following relaxed MCC optimization problem: 
\begin{align} 
\textbf{LP1} \quad\quad \quad &\min_{\{{x_K\}}}\quad \sum_{K\in \mathcal{K}(V)} w_K^+x_K+w_K^-(1-x_K) \quad \quad \quad  \nonumber\\
&\text{s.t.} \quad \quad x_K\in [0,1]  \quad \quad \quad \quad \quad \quad \text{(for all $K \in \mathcal{K}(V)$)} \nonumber \\
	&\quad \quad x_{K_3}\leq x_{K_1}+x_{K_2} \quad \quad \text{(for all $(K_1, K_2, K_3)\in \Upsilon)$}  \nonumber
\end{align}
where 
\begin{align} 
\Upsilon=&\{(K_1, K_2, K_3) \in \left[\mathcal{K}(V)\right]^3: K_1, K_2, K_3\;\text{are distinct}\nonumber\\
&\text{unordered k-tuples},\,K_1\cap K_2\neq\emptyset,\,K_3\subset K_1\cup K_2\}.  \nonumber 
\end{align}

Note that the constraints imposed on triples in $\Upsilon$ ensure that if two motifs share vertices and belong to the same cluster,
the additional motifs formed by the vertices also belong to the same cluster. 

The LP solutions are rounded according to Algorithm~\ref{alg:rounding1}, described below. The intuition behind the rounding algorithm is to use the fractional solutions of the LP $k$-tuple variables to perform rounding on pairs of variables. The reason for using different variables in the LP and in the rounding procedure is that
the LP constraints are harder to state and analyze via pairwise variables, while rounding is harder to perform via $k$-tuple variables as they incur complex codependencies. The key is to transition from $k$-tuples to pairs of variables by recording the ``best motif'' to which an edge belongs, and then using the corresponding fractional value of the motif variable to perform neighborhood growing via edge incorporation.
\begin{algorithm}[h] 
    \caption{Rounding procedure with $\alpha\leq\frac{1}{k}$.}
    \label{alg:rounding1}
    \begin{algorithmic}
     	\STATE{Let $S = V(G)$}
	\WHILE{$|S|\geq k$}
        \STATE{Choose an arbitrary pivot vertex $v$ in $S$}
        \STATE{For all $u\in S/\{v\}$, compute $y_{vu}=\min_{K\subset S: v,u\in K} x_{K}$}
        \STATE{Let $\mathcal{N}_{\alpha}(v)=\{u\in S/\{v\}: y_{vu}\leq \alpha\}$}
	\IF{$\sum_{j\in \mathcal{N}_{\alpha}(v)} y_{vu}>  \frac{\alpha}{2}|\mathcal{N}_{\alpha}(v)|$}
	\STATE{Output the singleton cluster $\{v\}$}
	\ELSE
	\STATE{Output the cluster $C=\mathcal{N}'_{\alpha}(v)$}
	\STATE{Let $S=S/C$}
	\ENDIF
        \ENDWHILE
	\STATE{Output all clusters $C$}
    \end{algorithmic}
\end{algorithm}

\begin{theorem}\label{thm:1}
Let $k$ be a constant size of a motif. For any $\alpha\leq \frac{1}{k}$ and the probability constraint $w_K^++w_K^-=1$ satisfied by every motif $K$ of size $k$, the LP coupled with the rounding procedure of Algorithm~\ref{alg:rounding1} provides a $\frac{2}{\alpha}$-approximate solution to the MCC problem.
\end{theorem}
\begin{proof} The proof is given in Appendix~\ref{app2}. \end{proof}

It is easy to verify that $|\Upsilon|=\sum_{i=k+1}^{2k-1}{|V|\choose i} \left[{i\choose k}{k\choose 2k-i}/2\right] \left[{i\choose k}-2\right]$. For constants $k$ such that $k\ll n$, $|\Upsilon|=\Theta(n^{2k-1})$. This indicates that the number of constraints in the LP grows exponentially with the size of the motif, which may lead to computational issues when the motifs are large. The next LP has a significantly smaller number of triangle constraints, reduced from $\Omega(n^{2k-1})$ to $\Omega(n^{3})$. In particular, this LP excludes a number of triangle inequalities as constrains. One cannot reduce the number of constraints below $\Theta(n^{k})$, as $\Theta(n^{k})$ variables are needed to represent all possible $k$-tuples. 

To describe the LP, we introduce some auxiliary variables. Let $z_{vu},$ $v,u\in V,$ denote the indicator of the event that a pair of vertices $v,u$ belong to different clusters 
(i.e., $z_{vu}=0$ if $v$ and $u$ belong to the same cluster, and $z_{vu}=0$ otherwise). By replacing the indicator variables $z_{vu} \in [0,1]$ and letting $x_{K} \in [0,1]$ as before, we arrive at the following LP problem formulation. 
\begin{align} 
\textbf{LP2} \quad\quad \quad &\min_{\{x_K\}, \{z_{vu}\}}\quad \sum_{K\in \mathcal{K}(V)} w_K^+x_K+w_K^-(1-x_K) \quad \quad \quad \quad \quad \quad \\
&\text{s.t.} \quad \quad x_K\geq z_{vu}   \quad\quad\quad\quad \quad\quad\quad \quad \quad \text{(for all $K \in \mathcal{K}(V)$ and $v,u\in K$)},  \label{LP2ineq1} \\
&		\quad \quad x_K\leq \frac{1}{k-1}\sum_{v,u\in K, v<u} z_{vu}, \quad  x_K\leq 1\quad \quad\quad \quad \quad \quad \text{(for all $K \in \mathcal{K}(V)$)},   \label{LP2ineq2} \\
&		\quad \quad z_{vu}\geq 0 \quad\quad\quad\quad\quad\quad\quad\quad\quad\quad\quad\quad\quad \quad \quad \quad \quad \quad \text{(for all $i,j \in V$)},  \nonumber \\
&	\quad \quad z_{v_2v_3}\leq z_{v_1v_2}+z_{v_1v_3} \quad\quad\quad \quad \text{(for all distinct vertices $v_1, v_2, v_3\in V)$}.\nonumber
\end{align}

A simple counting argument reveals that the number of constraints in the LP equals $\Theta({n\choose k}{k\choose 2}+{n\choose 3})$. 
Note that the inequalities~\eqref{LP2ineq1} and~\eqref{LP2ineq2} handle constraints on the $k$-tuples: Placing any pair of vertices in $K$ across clusters places $K$ across clusters, and placing $K$ across clusters causes placing at least $k-1$ many pairs of vertices across clusters. 
For the practically most relevant case $k=3$, the number of constraints in the above described optimization problem is roughly twice of that used in classical LP-based correlation clustering solvers~\cite{charikar2003clustering}. 

Algorithm~\ref{alg:rounding2} described the rounding procedure for the solution of \textbf{LP2}. In this case, the procedure reduces to the classical region growing method of~\cite{bansal2004correlation,charikar2003clustering}.
\begin{algorithm}[h]
    \caption{Rounding Procedure with parameters $\alpha,\beta\leq\frac{1}{k}$}
    \label{alg:rounding2}
    \begin{algorithmic}
     	\STATE{Let $S = V(G)$}
	\WHILE{$|S|\geq k$}
        \STATE{Choose an arbitrary pivot vertex $v$ in $S$}
        \STATE{Let $\mathcal{N}_{\alpha}(v)=\{u\in S/\{v\}: z_{vu}\leq \alpha\}$}
	\IF{$\sum_{u\in \mathcal{N}_{\alpha}(v)} z_{vu}>  \beta\alpha|\mathcal{N}_{\alpha}(v)|$}
	\STATE{Output the singleton cluster $\{v\}$}
	\ELSE
	\STATE{Output the cluster $C=\mathcal{N}'_{\alpha}(v)$}
	\STATE{Let $S=S/C$}
        \ENDIF
        \ENDWHILE
	\STATE{Output $S$}
    \end{algorithmic}
\end{algorithm}

\begin{theorem}\label{thm:2} Let $k$ be a constant size of a motif. For any $\alpha,\beta \leq \frac{1}{k}$ and the probability constraint $w_K^++w_K^-=1$ satisfied by every motif $K$ of size $k$, the LP coupled with the rounding procedure of Algorithm~\ref{alg:rounding2} provides a $\frac{1}{\alpha\beta}$-approximate solution to the MCC problem.
\end{theorem}
\begin{proof} The proof of the theorem is presented in Appendix~\ref{app3}. \end{proof}

Observe that the approximation guarantees of Theorem~\ref{alg:rounding2} are worse than those of Theorem~\ref{alg:rounding1}, which is the price paid for reducing the number of constraints. Furthermore, since the rounding procedure operates 
on pairs of vertices only and does not involve variables for $k$-tuples, it may be used for solving the MMCC problem as well. We outline the corresponding result in what follows.

let $S=\{k_1,k_2,...,k_p\}$ be the set of motif sizes of interest, and let $\mathcal{K}_t(V)$ be the set of all $k_t$-tuples of 
$V$. Using the same notation as in the MCC version of the problem, we may state the following LP relaxation for the MMCC problem:
\begin{align} 
\textbf{LP3}\quad\quad&\min_{\{x_K\},\{z_{uv}\}}\quad \sum_{t=1}^p \lambda_t\left[\sum_{K\in\mathcal{K}_t(V)} w_K^+x_K+w_K^-(1-x_K)\right] \quad \quad \quad \quad \quad \quad \nonumber \\
&\text{s.t.} \quad \quad x_K\geq z_{uv}   \quad \quad\quad\quad \quad  \text{(for all $1\leq t\leq p$, $K \in \mathcal{K}_t(V)$ and $u,v\in E(K)$)}, \nonumber \\
&		\quad \quad x_K\leq \frac{1}{|K|-1}\sum_{u,v\in K, u<v} z_{uv}, \quad  x_K\leq 1 \quad \quad  \text{(for all $1\leq t\leq p$, $K \in \mathcal{K}_t(V)$)},  \nonumber \\
&		\quad \quad z_{uv}\geq 0 \quad\quad\quad\quad\quad\quad\quad\quad\quad\quad\quad\quad\quad \quad \quad \quad \quad \quad \text{(for all $u,v \in V$)},   \nonumber\\
&	\quad \quad z_{u_2u_3}\leq z_{u_1u_2}+z_{u_1u_3} \quad\quad\quad \quad \text{(for all distinct vertices $u_1, u_2, u_3\in V)$}. 
\end{align}
The rounding method accompanying this LP is also described in Algorithm~\ref{alg:rounding2}, with the parameters $\alpha,\beta$ bounded from above by $\frac{1}{k^*}$, where $k^*=\max \,S=\max\{k_1,k_2,...,k_p\}$. 

\begin{corollary} For $\alpha,\beta \leq \frac{1}{k^*}$, and all motif weights satisfying the probability constraint $w_K^++w_K^-=1$, the rounded LP algorithm provides an $\frac{1}{\alpha\beta}$-approximate solution to the MMCC problem. \end{corollary}
\begin{proof} Note that the simplest way to prove this result is to focus on the largest motif only, and use the previously described MCC result. In particular, the stated result does not depend on the particular choices of the parameters $\lambda$ used.  \end{proof}

Still, one can derive more precise and stronger approximation guarantees by focusing on all motifs simultaneously, in which case the analysis becomes rather tedious and involved. For the special case of two motifs ($p=2$) with sizes $k_1=2$ and $k_2 =k$ respectively, we provide tighter approximation results in Theorem~\ref{twomotifs}. Here, both the parameters $\alpha,\beta$ depend on $\lambda$. The underlying derivations are relegated to Appendix~\ref{app4}.

\begin{theorem} \label{twomotifs}
Consider the MMCC problem with two types of motifs of sizes $k_1=2$ and $k_2 = k$. The objective function \textbf{LP3} may be rewritten as 
\begin{align*} 
\sum_{u, v\in V} \left[w_{uv}^+z_{uv}+w_{uv}^-(1-z_{uv})\right]+\lambda \sum_{K\in \mathcal{K}} \left[w_K^+x_K+w_K^-(1-x_K)\right]
\end{align*}
where $z_{uv}$ and $x_K$ are variables associated with pairs of vertices and $k$-tuples of vertices, and $\lambda$ is a parameter that can be tuned to balance the penalties induced by edges and motifs of size $k$. Let $r_0$ be a constant equal to
\begin{align*}
r_0 = \frac{k-2}{1+\lambda n^{k-1}}.
\end{align*}
Then, for any $\alpha\leq 1/k, \beta \leq 1/(k-r_0)$, and provided that the weights satisfy the probability constraint $w_K^++w_K^-=1$ for both $k$-tuples and edges (i.e., $w_{uv}^++w_{uv}^-=1$, the LP and rounding procedure of Algorithm~\ref{alg:rounding2} produce a $\frac{1}{\alpha\beta}$-approximate solution to the edge-motif MMCC problem.
\end{theorem}

\section{Numerical results for small social networks}

We evaluated our (M)MCC methods on two benchmark networks from~\cite{benson2015tensor}, which were originally tested using the method described in~\cite{benson2015tensor} (henceforth termed TSC), and on the well known Zachary karate club network~\cite{zachary1977information}. In all the experiments, we considered motifs of size $k=2$ and $k=3$ only. Hence, one of the motifs are edges and for the case $k=3$, the motif may be selected based on the particular application, as subsequently described. When solving MCC, we use the~\textbf{LP2} formulation as it contains fewer constraints than~\textbf{LP1}  and thus can be solved more efficiently. We then leverage Algorithm~\ref{alg:rounding2} for downstream rounding. When solving MMCC, we use a combination of~\textbf{LP3} and  Algorithm~\ref{alg:rounding2}. 


\subsection{Partitioning layered flow networks.} The first example is what we refer to as a \emph{layered flow network} (see Figure~\ref{fig:layeredflownet}). The information flow between two layers typically follows the same direction while feedback loops are primarily contained within a layer. The task is to detect the layers in the network. To perform the layer clustering, we assign the value $1$ to each weight $w_K$ corresponding to a directed $3$-cycle (i.e., triple $\{j_1,j_2,j_3\}$ with edges directed according to $j_1\rightarrow j_2$, $j_2\rightarrow j_3$, $j_3\rightarrow j_1$, or the reverse order), encouraging the corresponding triples to lie within a layer, while we assign a arbitrary weight in $[0.41, 0.48]$ to all other type of triples. The clustering results are shown in Figure~\ref{fig:layeredflownet}. Both MCC and the method of~\cite{benson2015tensor} produce similar clustering results, which identify the layers of the network. The only difference is observed for the node with label $3$. The MCC method emphasizes the feedback loops inside a layer, and hence node $3$ is placed in the same cluster as nodes $4,5,6,7$. The other method emphasizes the importance of the direction of information flow and thus the flow from node $3$ to node $1$ does not permit clustering nodes $3,4,5,6,7$ together.
\begin{figure}[th] 
\centering\includegraphics[trim={3cm 17.7cm 8cm 2.5cm},clip, width=0.6\columnwidth]{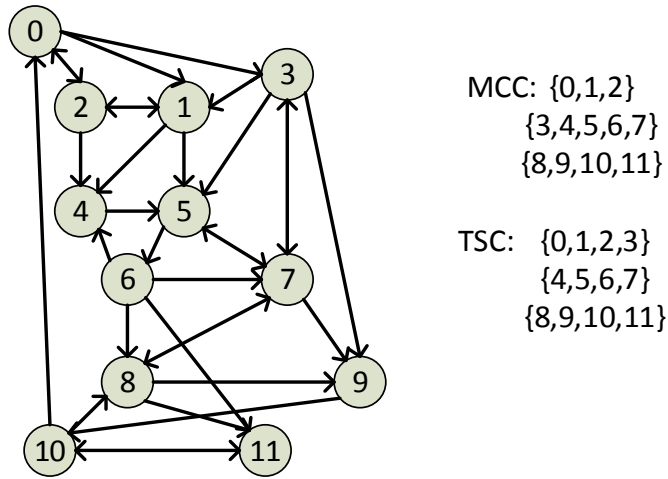} 
\caption{Example of a flow network, with layers detection performance of MCC and TSC. Left: The layered flow network; Right: The clustering results.}\label{fig:layeredflownet}  
\end{figure} 

\subsection{Anomaly detection.} Practical networks usually contain bidirectional edges, i.e., edges that allow both directions of traversal. A large number of these edges lie within directed $3$-cycles~\cite{benson2015tensor}. Hence, if a part of a network contains many directed $3$-cycles but very few bidirectional edges, it may be viewed as an anomaly. 

An illustrative example is shown in Figure~\ref{fig:anomalydetection}, in which the nodes labeled $0$-$5$ form an anomalous component which we wish to detect as it contains $8$ directed $3$-cycles without any bidirectional edges. The edges between nodes $6$-$21$ are generated according to a standard Erd\H{o}s-R\'enyi model with probability $0.25$ and to keep the figure simple, those edges were not plotted. Note that each of the nodes labeled $0-5$ has $4$ outgoing and $2$ incoming edges within the group of vertices containing $6-21$. There are $20$ directed triangles without bidirectional edges. 

To use our MCC method, we set the weights for the triangles without bidirectional edges to $1$, and those for other types of triangles to a value smaller than $0.42$. As the results shown in the Figure.~\ref{fig:anomalydetection} demonstrate, our method outperforms the TSC method in terms of detecting the anomaly. 

\begin{figure}[th] 
\centering
\includegraphics[trim={0.8cm 18.8cm 7cm 1cm},clip, width=0.6\columnwidth]{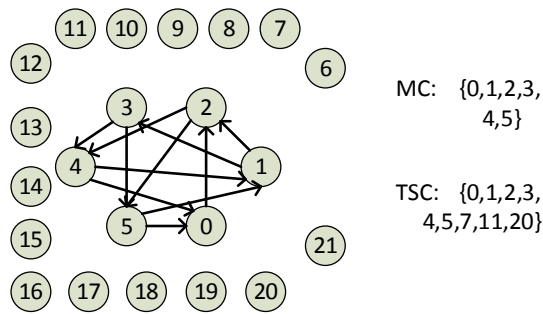} 
\caption{Anomaly detection in networks and clustering.}\label{fig:anomalydetection}  
\end{figure} 

\subsection{A benchmark social network: Zachary's karate club with two communities~\cite{zachary1977information}.} We also tested the performance of the CC, MCC  and MMCC methods
on the Zachary's karate club network. 
In the CC model, we assign weights to each pair of vertices weights depending on whether they are connected by an edge or not. For the MCC method, we focus on $3$-tuples and assign weights to the $3$-tuple weights according to whether their corresponding vertices form a triangle or a path. We use both triangles ($K_3$) and 3-paths ($P_3$) as motifs to ensure that nodes with very small degree can be clustered more accurately by examining their inclusion into important motifs involving vertices of large degree. The MMCC method uses both $2$-tuples and $3$-tuples. The weight assignments used in all these methods are listed Table~\ref{tab:karate-weights}. The result is shown in Figure.~\ref{tab:karate}. Although we  tested CC for a number of choices for the weights, we inevitably ended up with one clustering error, vertex $10$. This vertex is connected to $34$ in Cluster 1 and vertex $3$ in Cluster 2. On the other hand, the MCC and MMCC methods recovered the the ground truth clustering by taking into account the $K_3$ and $P_3$ motifs. The reason for this finding is that in social networks, vertices within a cluster typically connect to some central vertices in the same cluster (like vertex $34$ and vertex $1$). Hence, they form many triangles and $3$-paths containing the central vertices.

\begin{table}[h] 
  \caption{Weight assignments for the karate club network; $K_3$ stands for a triangle (complete graph on $3$ vertices), while $P_3$ denotes a path with three vertices.\\}
\label{tab:karate-weights}
   \centering
  \begin{tabular}{c c c c c c c}
    \toprule
    \midrule
   & \multicolumn{2}{c}{Edges}  & \multicolumn{3}{c}{Motifs}  &  \multirow{2}{*}{$\lambda$}  \\
        \cmidrule{1-6}
   Subgraphs   & Non-edges & Edges & Non-motifs & $K_3$ &  $P_3$ & \\
     \midrule
  CC & 0.47     & 1   & ---  &  ---  &  ---  & 0\\
  MCC & ---    & ---   & 0.49  &  1  &  2/3 & ---\\
  MMCC & 0.45    & 1   & 0.5  &  1  &  2/3 & 0.2\\
    \midrule
    \bottomrule
  \end{tabular}
\end{table}

%
%

\begin{figure}[h]
\centering 
   
   \includegraphics[trim={0cm 16cm 0cm 0cm}, clip, width=0.7\linewidth]{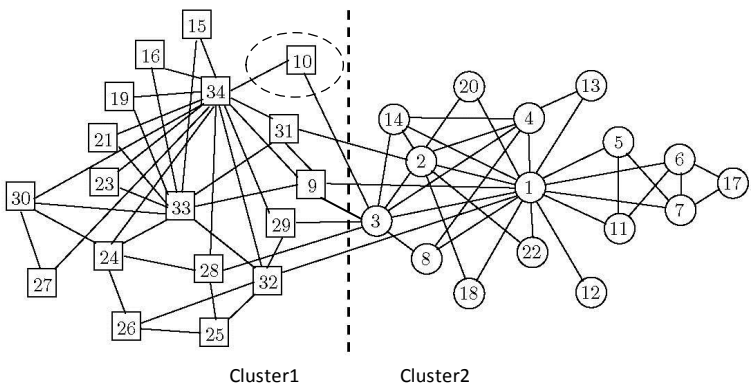}
   \caption{Clustering results for the CC, MCC and MMCC method performed on the Zachary's karate club network. Vertex $10$ is erroneously clustered by the CC method, but correctly clustered by both the MCC and MMCC methods.}
   \label{tab:karate}
%
\end{figure}
\textbf{Acknowledgment:} The authors gratefully acknowledge funding from the NSF Center of Science of Information for Science of Information, 487914, Purdue 4101-38050, NSF CCF grant 15-27636 and NIH grant 3U01CA198943-02S1.

\appendix
\section{Proof of Theorem~\ref{thm:hardness}} \label{app1}
To prove that the problem is in NP, we focus our attention on the case $(w^+_{K}, w^-_{K}) \in\{(1,0), (0,1)\}$. Since $w^+_{K}\in\{0,1\}$, as before, we refer to a triple $K$ with $w^+_{K}=1$ (respectively, $w^+_{K}=0$) as ``positive'' (respectively, ``negative''). We also use the term ``positive error'' to indicate that a positive triple is placed across clusters and ``negative error'' to indicate that a negative triple is placed within one cluster. 

Following the approach used to prove NP-hardness of CC~\cite{bansal2004correlation}, we use a reduction from the NP-complete Partition into Triangles~\cite{garey2002computers} problem. Given a (not necessarily complete) graph $G_{\triangle}=(V,E)$, containing $n$ vertices where $n$ is a multiple of 
$3$, the goal is to decide whether it can be partitioned into triangles. 

As the first step in our proof, we construct a graph $G^{w}$ that has the same vertex set as $G_{\triangle}$ and view triangles of $G^{w}$ as motifs. We set the weights of triples $G^{w}$ that correspond to triangles in $G_{\triangle}$  to $(1,0)$, and the weights of all other triples in $G^{w}$ to $(0,1)$. We solve the MCC problem over $G^{w}$ under the additional constraint that the size of each cluster is at most $3$. The existence of an efficient algorithm for solving this MCC would imply the existence of an efficient algorithm for partitioning $G_{\triangle}$ into triangles, a contradiction. As the original MCC algorithm does not necessarily generate clusters with bounded size $3$, in what follows we describe how to construct another graph, $H^{w}$, such that the triples-MCC algorithm applied on $H^{w}$ results in a bounded cluster-size run of MCC on $G^{w}$.

The basic idea behind our approach is to impose the constraint on the size of clusters in $G^{w}$ by 
adding vertices in $H^{w}$ for each triple in $G^{w}$, and then making the triples formed by the the newly added vertices positive and other triples negative. 
In this way, a cluster in the new graph $H^{w}$ with more than $3$ vertices in $G^{w}$ causes a large number of negative errors and hence 
cannot be part of an optimal clustering. 

We now describe how to construct a graph $H^w$ based on $G^w$. In addition to the vertices of $G^w$, for every triple $\{u_1,u_2,u_3\}$ in $G^w$, $H^w$ contains additional $n^5$ vertices, denoted by $C_{u_1u_2u_3}$. For simplicity of notation, write $C_{u_1u_2u_3} \cup \{u_1,u_2, u_3\} = C'_{u_1,u_2,u_3}$. Clearly, $H^w$ contains $n+n^5{n\choose 3}$ vertices. 
We classify the triples in $H^w$ into three types:
\begin{enumerate}
\item T-I triples: $\{u_1, u_2, u_3\}$, for all $u_1, u_2, u_3\in V(G^w)$. 
\item T-II triples: triples in $C'_{u_1,u_2,u_3}$ that are not T-I triples.
\item T-III triples: triples that are neither T-I triples nor T-II triples.
\end{enumerate}
The number of T-I triples is ${n \choose 3}$. As they are inherited from $G^w$, we keep their weights equal to those in $G^w$. The number of T-II triples equals ${n \choose 3}[{n^5+3 \choose 3}-1]$, and we assign the weights $(1,0)$ to them. The number of T-III triples equals ${n^5{n \choose 3} \choose 3} - {n \choose 3}{n^5+3 \choose 3}$, and we assign the weights $(0,1)$ to them.



Consider now a clustering $\mathcal{C}^*$ of $H^{w}$ of the following form:
\begin{enumerate}
\item There are ${n \choose 3}$ nonoverlapping clusters;
\item Each cluster contains one of the sets $C_{u_1u_2u_3}$ or one of the sets $C'_{u_1u_2u_3}$; 
\item Each vertex $u$ inherited from $V(G^w)$ lies in exactly one cluster.
\end{enumerate}
In the above clustering, there are no errors arising due to T-III triples, because all T-III triples are negative and $\mathcal{C}^*$ has property 2). The only errors arise from T-I triples and T-II triples. The number of errors induced by T-I triples is at most ${n \choose 3}$, while T-II triples errors in $\mathcal{C}^*$ may be grouped into two categories. First, a triple may have two vertices in $C_{u_1u_2u_3}$ and one vertex in $\{u_1,u_2,u_3\}$ that lies in another cluster. The number of this type of clustering errors is bounded from above by $n({n-1 \choose 2}-1){n^5 \choose 2}$. Second, a triple may have one vertex in $C_{u_1u_2u_3}$ and two vertices in $\{u_1,u_2,u_3\}$ that lie in another cluster. The number of this type of errors is upper bounded by ${n\choose 2}(n-3){n^5 \choose 1}$. Therefore, the total number of errors in $\mathcal{C}^*$ is bounded from above by 
\begin{align*}
n\left({n-1 \choose 2}-1\right){n^5 \choose 2}+{n\choose 2}(n-3){n^5 \choose 1}+{n \choose 3}\sim O(n^{13}).
\end{align*}

We may convert the clustering $\mathcal{C}^*$ into a partition $G^w$ based on the clustering of T-I triples. The clustering $\mathcal{C}^*$ essentially \emph{partitions the vertices of $G^w$} into clusters containing exactly three vertices. Our subsequent arguments aim to establish that the number of errors in a clustering that contains at least one cluster with at least four vertices from $V(G)$ must be larger than the number of errors induced by $\mathcal{C}^*$. 

For that purpose, consider another clustering of $H^w$, denoted by $\mathcal{C}'$. First, we show that in order for $\mathcal{C}'$ to have fewer errors than $\mathcal{C}^*$, the size of any cluster in $\mathcal{C}'$ must lie in the interval $[n^5-n^4, n^5+n^4]$. Suppose that on the contrary there exists a cluster containing more that $n^5+n^4$ vertices. Then, there are at least ${n^5 \choose 2}n^4\sim\Omega(n^{14})$ negative errors caused by placing T-III triples into this cluster. Furthermore, each cluster must contain at least $n^5-n^4$ vertices of a clique, otherwise there are at least ${n^5 \choose 2}n^4\sim\Omega(n^{14})$ positive errors generated by splitting the T-II triples. 
Second, note the each vertex in $V(G^w)$ belongs to ${n-1\choose 2}$ different triples of $G^w$. Since the size of each cluster of $\mathcal{C}'$ is smaller than $n^5+n^4$, for each vertex in $V(G)$, the number of negative errors caused by splitting the T-II triples that contains this vertex and two vertices from some $C_{u_1u_2u_3}$  is lower bounded by ${n^5 \choose 2}{n-1\choose 2}-{n^5 \choose 2}-{n^4 \choose 2}$. 

Assume now that there exists a cluster of $\mathcal{C}'$ that contains four vertices inherited from $V(G^w)$, say $\{u_1,u_2,u_3,u_4\}$. Then, as the size of the cluster is lower bounded by  $n^5-n^4$, from the pigeonhole principle it follows that there exists at least one vertex in $\{u_1,u_2,u_3,u_4\}$, say $j_1$, and at least $\frac{1}{4}(n^5-n^4)$ other vertices that do not lie in one of the sets $C_{u_1 u' u''}$ for some $u',u''\in v(G^w)$. Hence, the number of negative errors caused by T-III triples within this cluster is at least ${\frac{1}{4}(n^5-n^4) \choose 2}$. The total number of errors induced by such a clustering is therefore at least 
\begin{align*}
n{n^5 \choose 2}\left({n-1\choose 2}-1\right)-n{n^4 \choose 2}+{\frac{1}{4}(n^5-n^4) \choose 2},
\end{align*}
which is larger than the number of errors in the clustering $\mathcal{C}^*$, for $n$ sufficiently large. 
Therefore, the optimal triangle-clustering has to be of the form of $\mathcal{C}^*$, imposing a constraint on the size of clusters in $G^{w}$.

\section{Proof of Theorem~\ref{thm:1}} \label{app2}

Since we assume that the weights satisfy the probability constraint $w_{K}^+ + w_{K}^-=1$,  we will use $w_K$ to refer to $w_K^+$ and $1-w_K$ to refer to $w_K^-$.

Let $\mathcal{N}_{\alpha}(v)$ be the set defined in the rounding
procedure. If $\mathcal{N}_{\alpha}(v)\neq \emptyset$,
$\mathcal{N}_{\alpha}(v)$ contains at least $k-1$ elements, because if
$x_{K}\leq\alpha$ for some $k$-tuple $K$, then all its elements
(except possibly $v$) lie in $\mathcal{N}_{\alpha}(v)$. Let $\mathcal{N}'_{\alpha}(v)=\mathcal{N}_{\alpha}(v)\cup \{v\}$.
For convenience, we also define, given a pivot vertex $v$ and a $k$-tuple $K$ that contains $v$, $y_K = \sum_{u \in K/\{v\}} y_{vu}$. Furthermore, we let
$$K_{min}^{(uv)} = \arg\min_{K: \,u,v\in K} x_{K}.$$ Thus, by using the LP constraint and the definition of $y_{uv}$, we have 
\begin{align}\label{eq:yk-ineq}
\frac{1}{k-1}y_K  \leq x_K \leq \sum_{u\in K/\{v\}} x_{K_{min}^{(uv)}} = y_K.
\end{align}
Let $\kay_v$ be the set of all the $k$-tuples $K$ such that $K\subseteq \mathcal{N}'_{\alpha}(v),\, K\ni v$. When $v$ is a pivot vertex and $K\in \kay_v$, we know that 
\begin{align}\label{eq:yk-bound}
y_K\leq (k-1)\alpha\leq \frac{k-1}{k}.
\end{align}
The following proof often uses another form of the constraint in the underlying LP, i.e., $x_{K_1}\geq x_{K_3}-x_{K_2}$ for any $(K_1, K_2, K_3)\in\Upsilon$.

Next, we compare the rounding cost and the LP cost for different types of outputs of the algorithm. All possible cases and their corresponding approximation constants are listed in Table~\ref{tab:proofroadmap1}. 

\begin{table}[t]
\footnotesize
\begin{tabular}{c|c|c|c|c|c}
\hline
Output & Cost of & $K\ni v$ & Additional conditions &  Approx. constant & Case\# \\
\hline
\multirow{2}{*}{$\{v\}$} & splitting & yes & $K\cap[S/\mathcal{N}'_{\alpha}(v)]\neq \emptyset$ & $1/\alpha$ & Case 1 \\
\cline{2-6}
& splitting & yes & $K\cap[S/\mathcal{N}'_{\alpha}(v)]= \emptyset$ & $2/\alpha$ & Case 1 \\
\cline{1-6}
\multirow{10}{*}{$\mathcal{N}'_{\alpha}(v)$} & joint clustering & yes & --- & $1/[1-(k-1)\alpha]$ & Case 2.1 \\
\cline{2-6}
 & \multirow{2}{*}{joint clustering} & \multirow{2}{*}{no}  & Given $u\in \mathcal{N}_{\alpha}(v)$, $\exists K'\subseteq \mathcal{N}'_{\alpha}(v)$, & \multirow{2}{*}{$2/[2-(2k-1)\alpha]$} & \multirow{2}{*}{Case 2.1.1} \\
 & &  & $K'\ni v, u$, $x_{K'}\leq  \alpha/2$ & &\\
 \cline{2-6}
 & \multirow{2}{*}{joint clustering} & \multirow{2}{*}{no}  & Given $u\in \mathcal{N}_{\alpha}(v)$, $\forall K'\subseteq \mathcal{N}'_{\alpha}(v)$, & \multirow{2}{*}{$2/[2-(2k-1)\alpha]$} & \multirow{2}{*}{Case 2.1.2} \\
 & &  & $K'\ni v, u$, $x_{K'}>  \alpha/2$ & &  \\
  \cline{2-6}
 & splitting & yes  & --- & $1/\alpha$ & Case 2.2 \\
  \cline{2-6}
 & \multirow{2}{*}{joint clustering} & \multirow{2}{*}{no}  & $\exists K'$, $K'\ni v, x_{K'} \geq 1- \alpha/2$ & \multirow{2}{*}{$2/[2-(2k-1)\alpha]$} & \multirow{2}{*}{Case 2.2.1} \\
 & &  & $K'/\mathcal{N}'_{\alpha}(v) = K/\mathcal{N}'_{\alpha}(v)$  & &  \\
   \cline{2-6}
 & \multirow{2}{*}{splitting} & \multirow{2}{*}{no}  & $\forall K'$, $K'\ni v, x_{K'} < 1- \alpha/2$  & \multirow{2}{*}{$2/\alpha$} & \multirow{2}{*}{Case 2.2.2} \\
 & &  &  $K'/\mathcal{N}'_{\alpha}(v) = K/\mathcal{N}'_{\alpha}(v)$ &  &  \\
 \hline
\end{tabular}
\vspace{0.1in}
\caption{Overview of the different cases studied in the proof of Theorem~\ref{thm:1}. Output refers to the output of the algorithm; cost of splitting or joint clustering refers to the cost of splitting the $k$-tuple in $K$ or placing all of $K$ into into the output cluster. Additional conditions are specifics of the case under investigation.}
\label{tab:proofroadmap1}
\end{table}

\textbf{Case 1:} \emph{The output is the singleton cluster $\{v\}$}. The
clustering cost when outputting a singleton $\{v\}$ is $\sum_{K\subset
  \mathcal{K}(S): v\in K}w_{K}$ while the LP cost is $\sum_{K\subset
  \mathcal{K}(S): v\in K}(1-w_{K})(1-x_K)+w_{K}x_K$.

If $K\cap \left[S/\mathcal{N}'_{\alpha}(v)\right]\neq \emptyset$, we
have $x_{K}>\alpha$, so charging each such $k$-tuple $\frac{1}{\alpha}
w_{K}x_K$ times its LP-cost compensates for the
cluster-cost. Therefore, it suffices to consider the $k$-tuples $K\in \kay_v$. 
For $K\in \kay_v$, the LP cost
is bounded by
\begin{align*}
  &\sum_{K\in\kay_v}(1-w_{K})(1-x_K)+w_{K}x_K
  \geq \sum_{K \in \kay_v}(1-w_{K})(1-y_K)+w_{K}\frac{1}{k-1}y_K \\
  =&\sum_{K \in \kay_i}w_K\left[\frac{k}{k-1}y_K-1\right]+(1-y_K) 
  \geq \sum_{K \in \kay_v}\frac{1}{k-1}y_K
  \geq\frac{\alpha}{2}{|\mathcal{N}_{\alpha}(v)| \choose k-1},  
\end{align*}
where the first inequality is due to \eqref{eq:yk-ineq}, the second inequality is due to \eqref{eq:yk-bound} and $w_{K}\leq 1$, while the third inequality is due to the condition that the algorithm outputs a singleton cluster
$\{v\}$. Therefore, charging $\frac{2}{\alpha}$ for the $k$-tuple is
enough to compensate for the clustering cost.

\textbf{Case 2:} \emph{The output is the cluster $\mathcal{N}'_{\alpha}(v)$}. 

\textbf{Case 2.1:} First, consider the cost of the $k$-tuples inside the cluster. If $v\in K$, then we have $K\in \kay_v$ and thus $x_K\leq  y_K\leq
(k-1)\alpha$. So, charging $\frac{1}{1-(k-1)\alpha}$ for this tuple suffices to compensate the cluster-cost.

If $v\notin K$, order the vertices in $\mathcal{N}_{\alpha}(v)$ in
such a way that for any $u_1, u_2\in \mathcal{N}_{\alpha}(v)$, $u_1\prec u_2$
iff $y_{vu_1}< y_{vu_2}$ and assign an arbitrary order ($u_1\prec u_2$) when
the equality ($y_{vu_1}=y_{vu_2}$) holds.

For each vertex $u \in \mathcal{N}_{\alpha}(v)$, let $R_u=\{u' \in
\mathcal{N}_{\alpha}(v) \st u'\preceq u\}$, and let $\kay_v^{(u)}$ be the set of
$k$-tuples $K \in  \mathcal{N}_{\alpha}(v)$ such that $u$ is the largest vertex of $K$
according to $\prec$. Thus, if $K \in \kay_v^{(u)}$, then $u \in K$ and $K
\subset R_u $.

Note that because of the order, we have $\sum_{u'\in R_u} y_{vu'}\leq
\frac{\alpha}{2}|R_u|$.  Now for all $u\in \mathcal{N}_{\alpha}(v)$,
let us consider the total cost of the $k$-tuples in $R_u$. The
corresponding cluster-cost is $\sum_{K \in \kay_v^{(u)}}1-w_{K}$ while the LP
cost is $\sum_{K \in \kay_v^{(u)}}(1-x_{K})(1-w_{K})+x_{K}w_{K}$.

Next, let $\kay_v'$ be the set of $k$-tuples $K' \subset \mathcal{N}'_{\alpha}(v)$
with $v,u \in K'$.

\caze{2.1.1}{There is a clique $K' \in \kay_v'$ with
$x_{K'} \leq \frac{\alpha}{2}$.} 

Let $K^* = (K / \{u\}) \cup \{v\}$. Observe that $v \in K_v' \cap K^*$
and that $K \subset K_v' \cup K^*$.
Hence, the LP constraints imply that for all $K \in \kay_v^{(u)}$, we have
\[ x_{K}\leq x_{K'}+x_{K^*}\leq \frac{\alpha}{2}+(k-1)\alpha=\frac{(2k-1)\alpha}{2}. \]
So, charging $\frac{2}{2-(2k-1)\alpha}$ for each $k$-tuple in $ \kay_v^{(u)}$ is enough to compensate
the cluster-cost.

\caze{2.1.2}{For all $K' \in \kay_v'$, we have $x_{K'} >
  \frac{\alpha}{2}$.}  Let $K \in  \kay_v^{(u)}$, and let $\{u_1,
\ldots, u_{k-1}\}$ be the vertices in $K/\{u\}$. Let $y_K = \sum_{u_j \in
  K/\{g\}}y_{vu_j}$. For each $j \in \{1, \ldots, k-1\}$, let $K_j = (K /
\{u_j\}) \cup \{v\}$.  As each $K_j \in \kay_v'$, the LP constraints imply:
\begin{align*}
  &1-x_{K} \geq 1-\min_{j \in \{1, \ldots, k-1\}}\{x_{K_{min}^{(vu_j)}}+x_{K_j}\}\\
  &=1-\min_{j \in \{1, \ldots, k-1\}}\{y_{vu_j}+x_{K_j}\}\geq
  1-\left[\frac{1}{k-1}y_K+\frac{1}{k-1}\sum_{j=1}^{k-1}x_{K_j}\right]  
\end{align*}
and 
\[ x_{K}\geq  \max_{j \in \{1, \ldots, k-1\}}\{x_{K_j}-x_{K_{min}^{(vu_j)}}\} = \max_{j \in \{1, \ldots, k-1\}}\{x_{K_j}-y_{vu_j}\}\geq\frac{1}{k-1}\sum_{j=1}^{k-1}x_{K_j}-\frac{1}{k-1}y_K. \]
Let $\sigma = \sum_{j=1}^{k-1}x_{K_j}$. Manipulating these inequalities yields, for each $K \in \kay_v^{(u)}$,
\begin{equation}
  \label{eq:sigma}
  (1-w_{K})(1-x_{K})+w_{K}x_{K} \geq (1-w_{K})\left(1-\frac{2}{k-1}\sigma\right)
   -\frac{1}{k-1}y_K + \frac{1}{k-1}\sigma.  
\end{equation}
The LP constraints yield $x_{K_j} \leq (k-1)\alpha$ for each $j \in
\{1, \ldots, k-1\}$, since $i \in K_j$ for each $j$, by the same
argument used to establish inequality~\eqref{eq:yk-ineq}. Since each
$K_j \in \kay_v'$, we have $\alpha/2 \leq x_{K_j} \leq (k-1)\alpha$ for
each $j$, so that $\sigma \in [(k-1)\frac{\alpha}{2}, (k-1)^2\alpha]$.
The inequality~\eqref{eq:sigma} is linear in $\sigma$, so we study its
behavior when $\sigma$ is an endpoint of this interval.
When $\sigma = (k-1)\frac{\alpha}{2}$, we obtain
\begin{align*}
  (1-w_K)(1-x_K) + w_Kx_K &\geq (1-w_K)(1-\alpha) + \frac{\alpha}{2} - \frac{1}{k-1}y_K,
\end{align*}
and when $\sigma = (k-1)^2\alpha$, we obtain
\begin{align*}
  &(1-w_K)(1-x_K) + w_Kx_K \geq (1-w_K)(1 - 2(k-1)\alpha) + (k-1)\alpha - \frac{1}{k-1}y_K \\
  &= (1-w_K)(1-2(k-1)\alpha) + (k-\frac{3}{2})\alpha + \frac{\alpha}{2} - \frac{1}{k-1}y_K 
  \geq (1-w_K)(1-2(k-1)\alpha) + \\
  &(1-w_K)(k-\frac{3}{2})\alpha + \frac{\alpha}{2} - \frac{1}{k-1}y_K 
  = (1 - w_K)(1 - (k-\frac{1}{2})\alpha) + \frac{\alpha}{2} - \frac{1}{k-1}y_K.
\end{align*}
Since $k \geq 2$ we clearly have $(1-w_K)(1-(k-\frac{1}{2})\alpha) +
\frac{\alpha}{2} \leq (1-w_K)(1-\alpha) + \frac{\alpha}{2}$, so by
linearity, we also have
\begin{equation}
  \label{eq:egbound}
  (1-w_K)(1-x_K) + w_kx_K \geq (1-w_K)(1-(k-\frac{1}{2})\alpha) + \frac{\alpha}{2} - \frac{1}{k-1}y_K  
\end{equation}
for all $K \in  \kay_v^{(u)}$. Now, recall that $\sum_{u' \in R_u}y_{vu'} \leq \frac{\alpha}{2}|R_u|$; as
every vertex in $R_u$ appears in exactly ${\sizeof{R_u}-1 \choose k-2}$ $k$-tuples of $\kay_v^{(u)}$,
this implies that
\begin{align*}
  \frac{1}{k-1}\sum_{K \in \kay_v^{(u)}}y_K = \frac{1}{k-1}{\sizeof{R_u}-1 \choose k-2}\sum_{u' \in R_u}y_{vu'} 
  \leq \frac{\alpha}{2} {\sizeof{R_u} - 1 \choose k-2} \frac{\sizeof{R_u}}{k-1} 
  = \frac{\alpha}{2}{\sizeof{R_u} \choose k-1} = \frac{\alpha}{2}\sizeof{\kay_v^{(u)}}.
\end{align*}
Thus, summing inequality~\eqref{eq:egbound} over all tuples in $\kay_v^{(u)}$ yields the following lower
bound on the total LP-cost of these tuples:
\begin{align*}
  &\sum_{K \in \kay_v^{(u)}}[(1-w_K)(1-x_K) + w_kx_K] \geq \sum_{K \in \kay_v^{(u)}}\left[(1-w_K)(1-(k-\frac{1}{2})\alpha) + \frac{\alpha}{2} - \frac{1}{k-1}y_K\right] \\
  &\geq \sum_{K \in \kay_v^{(u)}}\left[(1-w_K)(1-(k-\frac{1}{2})\alpha)\right] = (1 - (k - \frac{1}{2}))\alpha \sum_{K \in \kay_v^{(u)}}(1-w_K).
\end{align*}
Therefore, charging $\frac{2}{2-(2k-1)\alpha}$ for each $k$-tuple in
$\kay_v^{(u)}$ suffices to compensate the cluster-cost.

\textbf{Case 2.2:} \emph{Compensating the cost of splitting a $k$-tuple.} Each tuple $K$ split during clustering incurs a cluster-cost of
$w_K$ and an LP-cost of $x_Kw_K + (1-x_K)(1-w_K)$.  First, suppose that
$K$ is a split $k$-tuple.
Since $K$ was split, $x_K > \alpha$, and charging $\frac{1}{\alpha}$ times the LP cost pays
for such a $K$.

Let $S' \subset S/\mathcal{N}'_{\alpha}(v)$ be such that $|S'|\leq k-1$. Furthermore, let $\kay_{v}^{(S')}$ be the set of split tuples $K$ such that $v\notin K$ and $K/\mathcal{N}'_{\alpha}(v)=S'$. According to the definition of $S'$, for any split tuple $K$, there is a corresponding $S'$. We show that the total cluster-cost of the
tuples in $\kay_{v}^{(S')}$ is at most a constant times their total LP-cost. To establish the claim, let
$\mathcal{S}_{\mathcal{N}}$ be the collection of all subsets $S_{\mathcal{N}} \subset \mathcal{N}_{\alpha}(v)$ with
$\sizeof{S_{\mathcal{N}}} = k-1-\sizeof{S'}$.

\caze{2.2.1}{There is some $S_{\mathcal{N}} \in \mathcal{S}_{\mathcal{N}}$ such that $x_{\{v\} \cup S_{\mathcal{N}}
    \cup S'} \geq 1 - \frac{\alpha}{2}$.} For each $K \in \kay_{v}^{(S')}$, let
$\tilde{S} = K \cap \mathcal{N}_{\alpha}(v)$, and take an arbitrary set $\bar{S} \subset \mathcal{N}_{\alpha}(v)/\tilde{S}$
with $\sizeof{\bar{S}} = k-1-\sizeof{\tilde{s}}$.  We have $x_{\{v\} \cup \tilde{S}
  \cup \bar{S}} \leq (k-1)\alpha$ and thus
\[
  x_K \geq x_{\{v\} \cup S' \cup S_{\mathcal{N}}} - x_{\{v\} \cup \tilde{S} \cup \bar{S}} \geq 1 - \frac{\alpha}{2} - (k-1)\alpha
  = \frac{2 - (2k-1)\alpha}{2},
\]
and in particular $x_K \geq \frac{2-(2k-1)\alpha}{2}$ for all $K \in
\kay_{v}^{(S')}$.  Thus, charging $\frac{2}{2-(2k-1)\alpha}$ times the LP-cost
to each $K \in\kay_{v}^{(S')}$ pays for the cluster-cost of all such edges.

\textbf{Case 2.2.2:} \emph{For all $S_{\mathcal{N}} \in \mathcal{S}_{\mathcal{N}}$, $x_{\{i\} \cup S_{\mathcal{N}} \cup S'} < 1 - \frac{\alpha}{2}$.}
Take any $K \in \kay_{v}^{(S')}$ and let $\tilde{S} = K \cap \mathcal{N}_{\alpha}(v)$. Suppose that $\tilde{S}=\{u_1, u_2, ..., u_{|\tilde{S}|}\}$. For each $u_j\in \tilde{S}$, let $K_j = (K / \{u_j\}) \cup \{v\}$. Note
that each tuple $K_j$ is a split tuple. We have:
\begin{gather*}
  1 - x_{K} \geq 1 - \min_{u_j \in \tilde{S}}[x_{K_{min}^{(vu_j)}}+ x_{K_j}] = 1 - \min_{u_j \in \tilde{S}}[y_{vu_j}+ x_{K_j}]
  \geq 1 - \frac{1}{\sizeof{\tilde{S}}}\sum_{u_j \in \tilde{S}}(y_{vu_j} + x_{K_j}); \\
  x_{K} \geq  \max_{u_j \in \tilde{S}}[x_{K_j} - x_{K_{min}^{(vu_j)}}]= \max_{u_j \in \tilde{S}}[x_{K_j} - y_{vu_j}]
  \geq \frac{1}{\sizeof{\tilde{S}}}\sum_{u_j \in \tilde{S}}(x_{K_j} - y_{vu_j}).
\end{gather*}
Let $\sigma_x = \sum_{u_j \in \tilde{S}}x_{K_j}$ and let $\sigma_y = \sum_{u_j \in \tilde{S}}y_{K_j}$.
The inequalities above yield the following lower bound on the LP-cost of $K$:
\begin{align}\label{ineq:lp1}
 (1 - w_K)(1 - x_K) + w_Kx_K \geq w_K\left[\frac{2}{\sizeof{\tilde{S}}}\sigma_x - 1\right] + 1 - \frac{1}{\sizeof{\tilde{S}}}(\sigma_x + \sigma_y).
\end{align}
We have $x_{K_j} \geq \alpha$ by definition and $x_{K_j}\leq 1-\frac{2}{\alpha}$ due to the assumptions made for this case. Thus, we have $\sigma_x \in [\alpha\sizeof{\tilde{S}}, (1-\frac{\alpha}{2})\sizeof{\tilde{S}}]$. As the lower bound in inequality~\eqref{ineq:lp1} is linear in $\sigma_x$, we study the behavior of the bound at the endpoints. When $\sigma_x =
\alpha\sizeof{\tilde{S}}$, we have
\begin{align*}
&(1-w_K)(1-x_K) + w_Kx_K \geq (2\alpha-1)w_K + 1 - \alpha - \frac{\sigma_y}{\sizeof{\tilde{S}}} \\
&\geq (2\alpha-1)w_K + (1-\frac{3\alpha}{2})w_K + \frac{\alpha}{2} - \frac{\sigma_y}{\sizeof{\tilde{S}}} \\
&= \frac{\alpha}{2}w_K + \frac{\alpha}{2} - \frac{\sigma_y}{\sizeof{\tilde{S}}}.
\end{align*}
Here, we used the fact that $\alpha < 2/3$. When $\sigma_x = (1-\frac{\alpha}{2})\sizeof{\tilde{S}}$, we obtain
\[ (1-w_K)(1-x_K) + w_Kx_K \geq (1-\alpha)w_K + \frac{\alpha}{2} - \frac{\sigma_y}{\sizeof{\tilde{S}}}. \]
Since $\alpha \leq 2/3$, we have $1-\alpha \geq \frac{\alpha}{2}$, so that the inequality
\begin{equation}
  \label{eq:sigmay}
  (1-w_K)(1-x_K) + w_Kx_K \geq \frac{\alpha}{2}w_K + \frac{\alpha}{2} - \frac{\sigma_y}{\sizeof{\tilde{S}}}
\end{equation}
holds for $\sigma_x$ at both endpoints of the interval, and thus holds for all $K \in \kay_v^{(S')}$. With $\tilde{S} = K \cap \mathcal{N}_{\alpha}(v)$ as before, we have $\sizeof{\tilde{S}} = k - \sizeof{S'}$ for all $K \in \kay_v^{(S')}$, and indeed
the map $K \mapsto (K \cap \mathcal{N}_{\alpha}(v))$ is a bijection from $\kay_v^{(S')}$ to ${\mathcal{N}_{\alpha}(v) \choose k - \sizeof{S'}}$. Since
each vertex of $\mathcal{N}_{\alpha}(v)$ lies in exactly ${\sizeof{\mathcal{N}_{\alpha}(v)}-1 \choose k - \sizeof{S'}-1}$ of the sets in ${\mathcal{N}_{\alpha}(v) \choose k - \sizeof{S'}}$,
we have
\begin{align*}
  &\sum_{K \in \kay_v^{(S')}}\frac{1}{\sizeof{\tilde{S}}}\sum_{u_j \in \tilde{S}} y_{vu_j} = \sum_{\tilde{S} \in {|\mathcal{N}_{\alpha}(i)| \choose k - \sizeof{S'}}}\frac{1}{k - \sizeof{S'}}\sum_{u_j \in \tilde{S}}y_{vu_j} = \frac{1}{k - \sizeof{S'}}{\sizeof{\mathcal{N}_{\alpha}(i)}-1 \choose k - \sizeof{S'}-1}\sum_{u\in \mathcal{N}_{\alpha}(v)} y_{vu} \\
  &\leq \frac{1}{k - \sizeof{S'}}{\sizeof{\mathcal{N}_{\alpha}(v)}-1 \choose k - \sizeof{S'}-1}\frac{\alpha\sizeof{\mathcal{N}_{\alpha}(v)}}{2}= \frac{\alpha}{2}{\sizeof{\mathcal{N}_{\alpha}(v)} \choose k - \sizeof{S'}} \qquad = \frac{\alpha}{2}\sizeof{\kay_v^{(S')}}.
\end{align*}
Thus, summing inequality~\eqref{eq:sigmay} over all tuples in $\kay_v^{(S')}$ yields the following lower bound on the total LP-cost of the underlying tuples:
\begin{align*}
  \sum_{K \in \kay_v^{(S')}}\left[(1-w_K)(1-x_K) + w_Kx_K\right] &\geq \sum_{K \in \kay_v^{(S')}}\left[\frac{\alpha}{2}w_K + \frac{\alpha}{2} - \frac{\sigma_y}{\sizeof{\tilde{S}}}\right] \geq \frac{\alpha}{2}\sum_{K \in \kay_v^{(S')}}w_K.
\end{align*}
Thus, charging each tuple in $\kay_v^{(S')}$ a factor of $\frac{2}{\alpha}$ times its LP-cost is enough
to pay for the cluster-cost.

In summary, if $\alpha= 1/k$ and we define
$c=\max\{\frac{1}{1-\alpha}, \frac{1}{1-(k-1)\alpha},
\frac{2}{2-(2k-1)\alpha}, \frac{2}{\alpha}\}=\frac{2}{\alpha}= 2k$,
then Algorithm~\ref{alg:rounding1} charges each $k$-tuple at most a factor of $2k$ times its LP.

\section{Proof of Theorem~\ref{thm:2}} \label{app3}

We continue to use the notation introduced in Appendix~\ref{app2}. In particular, we let $\mathcal{N}'_{\alpha}(v)=\mathcal{N}_{\alpha}(v)\cup \{v\}$ and let $\kay_v$ be the set of all $k$-tuples $K$ such that  $K\subseteq \mathcal{N}'_{\alpha}(v), \, K \ni v$. The following proof often uses some immediate consequences of the LP constraints; 
here we adopt the convention that $z_{uu} = 0$ for all $u \in V$:
\begin{enumerate}
\item $z_{u_1u_2}\geq z_{u_1u_3}-z_{u_2u_3}$ for any $u_1, u_2, u_3\in V$;
\item $x_{K}\geq \max_{uu'\in K} z_{uu'} \geq \max_{u,u' \in K}[z_{vu}-z_{vu'}]$, for any $v\in V$;
\item $x_{K} \leq \frac{1}{k-1}\sum_{u,u'\in K, u< u'}z_{uu'}\leq \frac{1}{k-1}\sum_{u,u' \in K, u<u'}(z_{vu}+z_{vu'})\leq \sum_{u\in K}z_{vu}$ for any $v\in V$.
\end{enumerate}
As before, we prove the approximation guarantees by comparing the rounding cost and the LP cost. An overview of the different cases encountered and the corresponding approximation constants is provided in Table~\ref{tab:proofroadmap2}. 

\begin{table}[t]
\footnotesize
\begin{tabular}{c|c|c|c|c|c}
\hline
Output & Cost of & $K\ni v$ & Additional conditions &  Approx. constant & Case\# \\
\hline
\multirow{2}{*}{$\{v\}$} & splitting & yes & $K\cap[S/\mathcal{N}'_{\alpha}(v)]\neq \emptyset$ & $1/\alpha$ & Case 1 \\
\cline{2-6}
& splitting & yes & $K\cap[S/\mathcal{N}'_{\alpha}(v)]= \emptyset$ & $1/(\alpha\beta)$ & Case 1 \\
\cline{1-6}
\multirow{6}{*}{$\mathcal{N}'_{\alpha}(v)$} & joint clustering & yes & --- & $1/[1-(k-1)\alpha]$ & Case 2.1 \\
\cline{2-6}
 &joint clustering& no  & Given $u\in \mathcal{N}_{\alpha}(v)$, $z_{vu}\leq \alpha\beta$, & $1/[1-k\alpha\beta]$ &Case 2.1.1 \\
 \cline{2-6}
 & joint clustering & no & Given $u\in \mathcal{N}_{\alpha}(v)$, $z_{vu}> \alpha\beta$, & $1/[1-(k-1)\alpha-\alpha\beta]$ & Case 2.1.2 \\
  \cline{2-6}
 & splitting & yes  & --- & $1/\alpha$ & Case 2.2 \\
  \cline{2-6}
 &splitting & no  & $\exists u \in K/\mathcal{N}_{\alpha}(v)$, $z_{vu}\geq (1+\beta)\alpha$ &$1/(\alpha\beta)$ & Case 2.2.1 \\
   \cline{2-6}
 & splitting & no  & $\forall u  \in K/\mathcal{N}_{\alpha}(v)$, $z_{vu}< (1+\beta)\alpha$  &$1/[\alpha(1-(k-1)\beta)]$ & Case 2.2.2 \\
 \hline
\end{tabular}
\vspace{0.1in}
\caption{Overview of the different cases studied in the proof of Theorem~\ref{thm:2}. Output refers to the output of the algorithm; cost of splitting or joint clustering refers to the cost of splitting the $k$-tuple in $K$ or placing all of $K$ into into the output cluster. Additional conditions are specifics of the case under investigation.}
\label{tab:proofroadmap2}
\end{table}

\textbf{Case 1:} \emph{The output is the singleton cluster $\{v\}$}. 
 The clustering cost when outputting a singleton $\{v\}$ is $\sum_{K\subset \mathcal{K}(S): v\in K}w_{K}$ while the LP cost is $\sum_{K\subset \mathcal{K}(S): i\in K}(1-w_{K})(1-x_K)+w_{K}x_K$. 

 If $K\cap[S/\mathcal{N}'_{\alpha}(v)]\neq \emptyset$, we have
 $x_{K}>\alpha$, so charging each such $k$-tuple $\frac{1}{\alpha}$
 times its LP-cost compensates for the cluster-cost. Therefore, it
 suffices to consider the $k$-tuples $K \in \kay_v$.
 
 For any $K \in \kay_v$, we have $\frac{1}{k-1}\sum_{u \in K/\{v\}}z_{vu}\leq x_{K}\leq \sum_{u \in K/\{v\}}z_{vu}$, where the inequalities are based on the LP constraints. By observing that $z_{vu}\leq \alpha$, we have the following bound on the LP cost of $K$:
\begin{align*}
(1-w_{K})(1-x_K)+w_{K}x_K
&\geq(1-w_{K})(1-\sum_{u \in K/\{v\}}z_{vu})+w_{K}\frac{1}{k-1}\sum_{u \in K/\{v\}}z_{vu} \\
&=w_K\left[\left(\frac{k}{k-1}\sum_{u\in K/\{v\}}z_{vu}\right)-1\right]+(1-\sum_{u \in K/\{v\}}z_{vu}).
\end{align*}
Since each $z_{vu}$ for $u \in K$ satisfies $z_{vu} \leq \alpha \leq 1/k$, the quantity
in square brackets is negative, so that $w_K \leq 1$ implies
\[
(1-w_{K})(1-x_K) + w_Kx_K \geq \frac{1}{k-1}\sum_{u \in K/\{v\}}z_{vu}.
\]
Summing over all $K \in \kay_v$, we see that
\[ \sum_{K \in \kay_v}[(1-w_K)(1-x_K) + w_Kx_K] \geq \sum_{K \in \kay_v}\sum_{u \in K/\{v\}}\frac{1}{k-1}z_{vu} \geq \alpha\beta{\sizeof{\mathcal{N}_{\alpha}(v)} \choose k-1}, \]
where the last inequality follows from the condition $\sum_{u \in \mathcal{N}_{\alpha}(v)}z_{vu} > \beta\alpha\sizeof{\mathcal{N}_{\alpha}(v)}$
that causes the algorithm to output $\{v\}$ as a singleton cluster.

Therefore, charging $\frac{1}{\alpha\beta}$ times the LP-cost to each
$k$-tuple in $\kay_v$ is enough to compensate for the total clustering cost
of the tuples in $\kay_v$.

\textbf{Case 2:} \emph{The output is the cluster $\mathcal{N}'_{\alpha}(v)$}.

\caze{2.1} \emph{First, consider the cost of the $k$-tuples inside the cluster.}
If $v\in K$, then we have $x_K\leq \sum_{u \in K/\{v\}}z_{vu}\leq (k-1)\alpha$, so charging
$\frac{1}{1-(k-1)\alpha}$ for this tuple suffices to compensate the
cluster-cost.

If $v\notin K$, order the vertices in $\mathcal{N}_{\alpha}(v)$ in
such a way that for any $u,u'\in \mathcal{N}_{\alpha}(v)$, $u\prec u'$
iff $z_{vu}< z_{vu'}$ and assign an arbitrary order ($u\prec u'$) when
the equality ($z_{vu}=z_{vu'}$) holds.

For each vertex $u \in \mathcal{N}_{\alpha}(v)$, let $R_u=\{u' \in
\mathcal{N}_{\alpha}(v) \st u'\preceq u\}$, and let $\kay_v^{(u)}$ be the set of
cliques $K \in \kay_v^{(u)}$ such that $l$ is the largest vertex of $K$
according to $\prec$. Thus, if $K \in \kay_v^{(u)}$, then $u \in \kay_v^{(u)}$ and $K
\subset R_u $.

Note that because of the order, we have $\sum_{u'\in R_u} z_{vu'}\leq
\alpha\beta|R_u|$.  Fix some $u \in \mathcal{N}_{\alpha}(v)$, and consider the
total cost of the $k$-tuples in $\kay_v^{(u)}$. The corresponding cluster-cost
is $\sum_{K \in \kay_v^{(u)}}1-w_{K}$ while the LP cost is $\sum_{K \in
  \kay_v^{(u)}}(1-x_{K})(1-w_{K})+x_{K}w_{K}$.

Let $\kay_v'$ be the set of $k$-tuples $K \subset \mathcal{N}'_{\alpha}(v)$
with $v,u \in K$.

\textbf{Case 2.1.1:} \emph{$z_{vu} \leq \beta\alpha$}. In this case, for each
$K \in \kay_v^{(u)}$, we have
\[ x_K \leq \sum_{u' \in K}z_{vu'} \leq kz_{vu} \leq k\beta\alpha, \] so
that charging $\frac{1}{1-k\beta\alpha}$ times the LP-cost to each
$k$-tuple in $\kay_v^{(u)}$ suffices to pay for the cluster cost of all such
tuples.

\textbf{Case 2.1.2:} \emph{$z_{vu} > \beta\alpha$}. In this case, by using $x_K \leq \sum_{u \in K}z_{vu}$, we have $1 - x_K \geq 1 - \sum_{u \in K}z_{vu}$. Furthermore,
\[ x_K \geq z_{vu} - \min_{u'\in K/\{u\}}z_{vu'} \geq z_{vu} - \frac{1}{k-1}\sum_{u' \in K/\{u\}}z_{vu'}. \]
Letting $\sigma = \sum_{u' \in K/\{u\}}z_{vu'}$ so that $1-x_K \geq 1 - z_{vu} - \sigma$, we have
the following lower bound on the LP-cost of $K$:
\begin{align*}
  (1-w_K)(1-x_K) + w_Kx_K &\geq (1-w_K)(1 - z_{uv} - \sigma) + w_K(z_{uv} - \frac{1}{k-1}\sigma) \\
  &= (1-w_K)(1 - 2z_{uv} - \frac{k-2}{k-1}\sigma) + z_{uv} - \frac{1}{k-1}\sigma. 
\end{align*}
Now, summing over all $K \in \kay_v^{(u)}$ and using the inequality $\sum_{K \in \kay_v^{(u)}}\frac{1}{k-1}\sum_{u' \in K/\{u\}}z_{vu'} \leq \sizeof{\kay_v^{(u)}}\beta\alpha$ yields the following lower bound on the total LP-cost of the $k$-tuples in $\kay_v^{(u)}$:
\begin{align*}
  &\sum_{K \in \kay_v^{(u)}}[(1-w_k)(1-x_K) + w_Kx_K] \geq \sum_{K \in \kay_v^{(u)}}[(1-w_K)(1 - 2z_{uv} - \frac{k-2}{k-1}\sigma)+ z_{uv} - \beta\alpha] \\
  &\geq \sum_{K \in \kay_v^{(u)}}[(1-w_K)(1-z_{uv}-\frac{k-2}{k-1}\sigma - \beta\alpha)]\geq \sum_{K \in \kay_v^{(u)}}\left[(1-w_K)[1-(k-1)\alpha - \beta\alpha]\right].
\end{align*}
Thus, charging each k-tuple in $\kay_v^{(u)}$ a factor of $\frac{1}{1-(k-1)\alpha-\beta\alpha}$ times its LP-cost
pays for the cluster-cost of all k-tuples in $\kay_v^{(u)}$.

\textbf{Case 2.2:} \emph{The cost of splitting $k$-tuples across clusters.}  Again, we refer to such tuples as split tuples. Each split tuple $K$ incurs a cluster-cost of $w_K$ and an
LP-cost of $x_Kw_K + (1-x_K)(1-w_K)$. First, suppose that $K$ is a split $k$-tuple with $v
\in K$. Since $K$ is split, there is $u' \in K/\mathcal{N}'_{\alpha}(v)$ and thus we have $x_K \geq z_{vu'}
> \alpha$, so charging $\frac{1}{\alpha}$ times the LP cost pays for
such $K$. We still must pay for the split tuples $K$ with $v \notin K$.

Let $S' \subset S/\mathcal{N}'_{\alpha}(v)$ be such that $|S'|\leq k-1$. Furthermore, let $\kay_{v}^{(S')}$ denote the set of split tuples $K$ such that $v\notin K$ and $K/\mathcal{N}'_{\alpha}(v)=S'$. According to the definition of $S'$, for any split tuple $K$, there is a corresponding $S'$.  We show that the total
cluster-cost of the tuples in $\kay_{v}^{(S')}$ is at most a constant time their
total LP-cost.

\textbf{Case 2.2.1:} \emph{There exists a vertex $u \in S'$ such that $z_{vu} \geq (1+\beta)\alpha$}.
In this case, for every $K \in \kay_v^{(S')}$, we can take some arbitrary $u' \in K \cap \mathcal{N}_{\alpha}(v)$ and obtain
\[ x_K \geq z_{vu} - z_{vu'} \geq \beta\alpha, \]
since $u' \in \mathcal{N}_{\alpha}(v)$ implies $z_{vu} \leq \alpha$. Thus, in this case,
charging $\frac{1}{\alpha\beta}$ times the LP-cost of each tuple in $\kay_v^{(S')}$ pays for the
cluster-cost of all tuples in $\kay_v^{(S')}$.

\textbf{Case 2.2.2:} \emph{For all $u\in S'$, $z_{vu} \leq (1+\beta)\alpha$}. Consider any $K \in \kay_v^{(S')}$. Let $\tilde{S} = K \cap \mathcal{N}_{\alpha}(v)$, and $\sigma^{'}_{S'} =
\sum_{u \in S'}z_{vu},$ $\sigma^{''}_{\tilde{S}} = \sum_{u \in
  \tilde{S}}z_{vu}$. We have the following bounds:
\begin{gather*}
 1 - x_{K} \geq 1 - \sum_{u \in K}z_{vu} = 1 - (\sum_{u \in S'}z_{vu} + \sum_{u \in \tilde{S}}z_{vu}) = 1-(\sigma^{'}_{S'}+\sigma^{''}_{\tilde{S}}), \\
 x_K \geq \max_{u,u' \in K}[z_{vu} - z_{vu'}] \geq \max_{u \in S',\ u' \in \tilde{S}}[z_{vu} - z_{vu'}] \geq \frac{1}{\sizeof{S'}}\sigma^{'}_{S'} - \frac{1}{\sizeof{\tilde{S}}}\sigma^{''}_{\tilde{S}}.
\end{gather*}
Combining these bounds yields the following lower bound on the LP-cost of $K$.
\begin{align}
  (1-w_K)(1-x_K) + w_Kx_K &\geq (1-w_K)(1 - \sigma^{'}_{S'}- \sigma^{''}_{\tilde{S}}) + w_K\left(\frac{\sigma^{'}_{S'}}{\sizeof{S'}} - \frac{\sigma^{''}_{\tilde{S}}}{\sizeof{\tilde{S}}}\right) \label{eq:ksquare-bound}\\
  &= w_K\left[\frac{\sizeof{S'}+1}{\sizeof{S'}}\sigma^{'}_{S'} + \frac{\sizeof{\tilde{S}}-1}{\sizeof{\tilde{S}}}\sigma^{''}_{\tilde{S}} - 1\right] + 1 - \sigma^{'}_{S'} - \sigma^{''}_{\tilde{S}}\nonumber.
\end{align}

The map $K \mapsto (K \cap \mathcal{N}_{\alpha}(v))$ induces a bijection between $\kay_v^{(S')}$ and ${\mathcal{N}_{\alpha}(i) \choose k - |S'|}$. By using $z_{vu'}\leq \alpha\beta$ for $u' \in  K\cap\mathcal{N}_{\alpha}(v)$, we have 
\begin{align}\label{eq:sigmaIIK}
& \sum_{K \in \kay_v^{(S')}}\sigma^{''}_{\tilde{S}} = \sum_{K \in \kay_v^{(S')}}\sum_{u'\in K\cap\mathcal{N}_{\alpha}(v)} z_{vu'}  \\
&= \frac{k-|S'|}{\sizeof{\mathcal{N}_{\alpha}(v)}}{\sizeof{\mathcal{N}_{\alpha}(v)} \choose (k-|S'|)} \sum_{u'\in \mathcal{N}_{\alpha}(v)} z_{vu'}  \nonumber
\leq  \sum_{K \in \kay_v^{(S')}} \alpha\beta\sizeof{K\cap\mathcal{N}_{\alpha}(v)}. \nonumber
\end{align}
Since $\alpha, \beta \leq 1/k$, we also have
\begin{align}\label{eq:1minussigmaI}
1 - \sigma^{'}_{K} - \alpha\beta \sizeof{\tilde{S}} \geq 1 - \sizeof{S'}(1+\beta)\alpha - \sizeof{\tilde{S}}\beta\alpha \geq 1 - (k-1)(1+\beta)\alpha - \beta\alpha \geq 0.
\end{align}
Therefore, summing inequality~\eqref{eq:ksquare-bound} over all $K \in \kay_v^{(S')}$ gives the following
lower bound on the total LP-cost of all tuples in $\kay_v^{(S')}$:
\begin{align*}
  &\sum_{K \in \kay_v^{(S')}}[(1-w_K)(1-x_K) + w_Kx_K]  \\
   &\geq \sum_{K \in \kay_v^{(S')}}\left( w_K\left[\frac{\sizeof{S'}+1}{\sizeof{S'}}\sigma^{'}_{S'} + \frac{\sizeof{\tilde{S}}-1}{\sizeof{\tilde{S}}}\sigma^{''}_{\tilde{S}} - 1\right] + 1 - \sigma^{'}_{S'} - \sigma^{''}_{\tilde{S}}\right)\\
   & \geq \sum_{K \in \kay_v^{(S')}}\left( w_K\left[\frac{\sizeof{S'}+1}{\sizeof{S'}}\sigma^{'}_{S'} + \frac{\sizeof{\tilde{S}}-1}{\sizeof{\tilde{S}}}\sigma^{''}_{\tilde{S}} - 1\right] + 1 - \sigma^{'}_{S'} - \alpha\beta|\tilde{S}|\right) \\
& \geq \sum_{K \in \kay_v^{(S')}}w_K\left[\frac{\sizeof{S'}+1}{\sizeof{S'}}\sigma^{'}_{S'} + \frac{\sizeof{\tilde{S}}-1}{\sizeof{\tilde{S}}}\sigma^{''}_{\tilde{S}} - 1 + 1 - \sigma^{'}_{S'} - \alpha\beta|\tilde{S}|\right]\\
   &\geq \sum_{K \in \kay_v^{(S')}}w_K\left[\frac{1}{\sizeof{S'}}\sigma^{'}_{S'} + \frac{\sizeof{\tilde{S}} - 1}{\sizeof{\tilde{S}}}\sigma^{''}_{\tilde{S}}- \alpha\beta|\tilde{S}|  \right] \geq \sum_{K \in \kay_v^{(S')}}w_K[\alpha + 0 - (k-1)\alpha\beta],
\end{align*}
where the second inequality is due to~\eqref{eq:sigmaIIK} and the third inequality follows from~\eqref{eq:1minussigmaI} and $w_K\leq 1$.

Therefore, charging a factor of $\frac{1}{\alpha(1-(k-1)\beta)}$ times the LP-cost for each tuple in $\kay_v^{(S')}$ pays for the cluster-cost of all tuples in $\kay_v^{(S')}$.
 
In summary, if $\alpha, \beta \leq 1/k$, then charging each tuple a factor of $c$ times its LP cost,
where
\begin{align*}
c=\max\{\frac{1}{\beta\alpha}, \frac{1}{1-(k-1)\alpha}, \frac{1}{1-(k-1)\alpha-\beta\alpha},\frac{1}{\alpha[1-(k-1)\beta]}\} = \frac{1}{\alpha\beta},
\end{align*}
suffices to compensate the cluster-cost of all tuples.

\section{Proof of Theorem~\ref{twomotifs}} \label{app4}

For the MMCC problem, the proof of Theorem~\ref{thm:2} (Appendix~\ref{app3}) may be generalized by independently handling tuples of fixed sizes. However, to obtain a tighter approximate constant then the one presented in Theorem~\ref{twomotifs}, we next show how to modify the corresponding analysis for Case 2.2.2. 

The analysis of Case 2.2.2 for mixed motifs proceeds as follows. Define $S^* = \{u \in S/\mathcal{N}'_{\alpha}(v), z_{vu} \leq (1+\beta)\alpha\}$ and $\bar{\sigma} = \frac{1}{|S^*|}\sum_{u\in S^*} z_{vu}\leq (1+\beta)\alpha$. For $S' \subseteq S^*$ of size $|S'|\leq k-1$, and for all $u\in S'$, it holds that $z_{vu}< (1+\beta)\alpha$. 

Let $\kay_v^{(S')}$ be the set of all $k$-tuples $K$ such that $K/\mathcal{N}'_{\alpha}(v) = S'$ and $v\notin K$. We need to find a constant $c$ such that 
\begin{align*}
&\sum_{u' \in \mathcal{N}_{\alpha}(v)}\sum_{u\in S^*}w_{u'u}+ \lambda\sum_{S' \subseteq S^*}\sum_{K\in \kay_v^{(S')} } w_{K} \\
&\leq c \left\{\sum_{u' \in \mathcal{N}_{\alpha}(v)}\sum_{u\in S^*} [w_{u'u}z_{u'u}+  (1-w_{u'u})(1-z_{u'u})] \right.\\
&\left. + \lambda\sum_{S' \subseteq S^*}\sum_{K\in \kay_v^{(S')} } [w_{K}x_{K}+  (1-w_{K})(1-x_{K})]\right\}.
\end{align*}

Recall that $\tilde{S} = K/S$, and that $\sigma^{'}_{S'} = \sum_{u \in S'}z_{vu}$ and $\sigma^{''}_{\tilde{S}} = \sum_{u' \in \tilde{S}}z_{vu'}$. Using the same method as the one outlined in the derivations of~\eqref{eq:ksquare-bound} and~\eqref{eq:sigmaIIK}, and observing that $\sigma^{''}_{\tilde{S}} \geq 0$, we obtain
\begin{align} \label{ineq:kmotifpart}
&\sum_{S' \subseteq S^*}\sum_{K\in \kay_v^{(S')} } [w_{K}x_{K}+  (1-w_{K})(1-x_{K})]  \\ \nonumber
&\geq  \sum_{S' \subseteq S^*}\sum_{K\in \kay_v^{(S')} } \left[w_K\left(\frac{\sizeof{S'}+1}{\sizeof{S'}}\sigma^{'}_{S'} - 1\right) + 1 - \sigma^{'}_{S'} - \alpha\beta|\tilde{S}|\right]\\ \nonumber
&\geq \sum_{t = 1}^{k-1}  \sum_{S' \subseteq S^*, |S'| = t}\left[(1 - \sigma^{'}_{S'} - \alpha\beta|\tilde{S}|)|\kay_v^{(S')}| +\left(\frac{t+1}{t}\sigma^{'}_{S'} - 1\right) \sum_{K\in \kay_v^{(S')} } w_K \right], 
\end{align}
where the sum of the coefficients in front of the term $\alpha\beta$ equals $-\sum_{t = 1}^{k-1}(k-t){|\mathcal{N}_{\alpha}(v)| \choose k-t}{|S^*| \choose t}$.

Using the same approach as for the derivations when $k=2$, we have 
\begin{align} \label{ineq:2motifpart}
&\sum_{u\in S^*} \sum_{u' \in \mathcal{N}_{\alpha}(v)} [w_{u'u}z_{u'u}+  (1-w_{u'u})(1-z_{u'u})]  \\
&\geq   \sum_{u\in S^*} \left[(1 - z_{vu} - \alpha\beta)| \mathcal{N}_{\alpha}(v)| + \left(2z_{vu} - 1\right) \sum_{u' \in \mathcal{N}_{\alpha}(v)} w_{uu'}\right], 
\end{align}
where the sum of the coefficients in front of the term $\alpha\beta$ equals $-\sizeof{\mathcal{N}_{\alpha}(v)}\sizeof{S^*}$.

Next, define two constants $r$ and $r'$ based on
\begin{align*}
r= \frac{(k-2)\sizeof{\mathcal{N}_{\alpha}(v)}\sizeof{S^*}}{\sizeof{\mathcal{N}_{\alpha}(v)}\sizeof{S^*}+\lambda\sum_{t=1}^{k-1}(k-t){|\mathcal{N}_{\alpha}(v)| \choose k-t}{|S^*| \choose t}},\, r' = (k-2) - r,
\end{align*}
so that they satisfy
\begin{align*}
|\mathcal{N}_{\alpha}(v)||S^*| r' = \lambda\sum_{t=1}^{k-1}(k-t){|\mathcal{N}_{\alpha}(v)| \choose k-t}{|S^*| \choose t} r.
\end{align*}
By choosing $\alpha \leq \frac{1}{k}$ and $\beta \leq \frac{1}{k-r}$, we can verify that, for $1\leq t \leq k-1$, any $S'\subseteq S^*$, $|S'| = t$, and $u\in S^*$,
\begin{align} \nonumber
1-\sigma^{'}_{S'} - (k-t- r)\alpha\beta&\geq 1-t\alpha(1+\beta)-(k-t- r)\alpha\beta \\
&\geq 1-(k-1)\alpha(1+\beta)-(1-r)\alpha\beta \geq 0,  \label{ineq:nonnega1}\\
1- z_{vu}- (1+ r')\alpha\beta&\geq 1-\alpha(1+\beta)-(k-1- r)\alpha\beta \geq 0.  \label{ineq:nonnega2}
\end{align}

Combining inequalities~\eqref{ineq:kmotifpart} and~\eqref{ineq:2motifpart} and inserting $r$ and $r'$ into the expressions, we obtain 
\begin{align*}
&\sum_{u\in S^*} \sum_{u' \in \mathcal{N}_{\alpha}(v)} [w_{u'u}z_{u'u}+  (1-w_{u'u})(1-z_{u'u})] \\
&+ \lambda\sum_{S' \subseteq S^*}\sum_{K\in \kay_v^{(S')} } [w_{K}x_{K}+  (1-w_{K})(1-x_{K})]  \\
&\geq  \sum_{u\in S^*} \left[(1 - z_{vu} - (1+r')\alpha\beta)| \mathcal{N}_{\alpha}(v)| + \left(2z_{vu} - 1\right) \sum_{u' \in \mathcal{N}_{\alpha}(v)} w_{uu'}\right] \\
& +  \lambda\sum_{t = 1}^{k-1}  \sum_{S' \subseteq S^*, |S'| = t}\left[(1 - \sigma^{'}_{S'} - (k-t-r)\alpha\beta)|\kay_v^{(S')}| +\left(\frac{t+1}{t}\sigma^{'}_{S'} - 1\right) \sum_{K\in \kay_v^{(S')} } w_K \right]  \\
& \geq  \sum_{u\in S^*}\left\{ \left[z_{vu} - (1+r')\alpha\beta\right] \sum_{u' \in \mathcal{N}_{\alpha}(v)} w_{uu'}\right\} \\
& + \lambda\sum_{t = 1}^{k-1}  \sum_{S' \subseteq S^*, |S'| = t}\left\{\left[\frac{1}{t}\sigma^{'}_{S'} - (k-t-r) \alpha\beta\right] \sum_{K\in \kay_v^{(S')} } w_K \right\} \\
& \geq \min\{\alpha - (1+r')\alpha\beta, \alpha -  (k-1-r) \alpha\beta\}\sum_{u' \in \mathcal{N}_{\alpha}(v)}\sum_{u\in S^*}w_{u'u}+ \lambda\sum_{S' \subseteq S^*}\sum_{K\in \kay_v^{(S')} } w_{K},
\end{align*} 
where the second inequality is due to inequalities \eqref{ineq:nonnega1} and \eqref{ineq:nonnega2}, and $w_{uu'}, w_{K}\leq 1$.

Therefore, charging a factor of $\min\{\alpha - (1+r')\alpha\beta, \alpha -  (k-1-r) \alpha\beta\} = \alpha - (k-1-r)\alpha\beta$ times the LP-cost for all pairs $(u,u')$ such that $u'\in \mathcal{N}_{\alpha}(v)$ and $u\in S^*$, and for all $k$-tuples $K$ such that $K/\mathcal{N}_{\alpha}(v) \subseteq S^*$ compensates for splitting all such pairs and $k$-tuples during clustering.
 
Combining all cases described in Table~\eqref{tab:proofroadmap2} shows that if $\alpha\leq 1/k, \beta \leq 1/(k - r)$, then charging each pair and $k$-tuple a factor of $c$ times its LP cost,
where
\begin{align*}
c=\max\{\frac{1}{\beta\alpha}, \frac{1}{1-(k-1)\alpha}, \frac{1}{1-(k-1)\alpha-\beta\alpha},\frac{1}{\alpha[1-(k-1-r)\beta]}\} = \frac{1}{\alpha\beta} \geq k(k - r),
\end{align*}
suffices to compensate the cluster-cost of all pairs and tuples.

Note that, however, $r$ depends on $|S^*|$ and $\mathcal{N}_{\alpha}(v)$ and these values are not known a priori and they may change over different iterations. Hence, we need to find a universal lower bound for $r$. Since $|S^*|+|\mathcal{N}_{\alpha}(v)|\leq n$, a simple bound of the form may be obtained according to
\begin{align*}
r\geq \frac{(k-2)\sizeof{\mathcal{N}_{\alpha}(v)}}{\sizeof{\mathcal{N}_{\alpha}(v)}+ \lambda \sum_{t=1}^{k-1}\sizeof{\mathcal{N}_{\alpha}(v)} {\sizeof{\mathcal{N}_{\alpha}(v)}-1 \choose k-t-1}{\sizeof{S^*} \choose t}} \geq \frac{k-2}{1+ \lambda n^{k-1}} =r_0.
\end{align*}
Therefore, if $\alpha\leq 1/k, \beta \leq 1/(k-r_0)$, one can achieve the constant approximation factor $c=1/\alpha\beta$. 

%
%
%

\bibliographystyle{plain}
\bibliography{motif-siopt.bib}

\end{document}